\documentclass[pdflatex,referee,sn-mathphys-num]{sn-jnl}

\usepackage{graphicx}%
\usepackage{multirow}%
\usepackage{amsmath,amssymb,amsfonts}%
\usepackage{amsthm}%
\usepackage{mathrsfs}%
\usepackage[title]{appendix}%
\usepackage[x11names]{xcolor}%
\usepackage{textcomp}%
\usepackage{manyfoot}%
\usepackage{booktabs}%
\usepackage{algorithm}%
\usepackage{algorithmicx}%
\usepackage{algpseudocode}%
\usepackage{listings}%
\usepackage[acronym,nohypertypes={acronym}]{glossaries}%
\usepackage[english]{babel}
\usepackage{mathtools}
\usepackage{subcaption}
\usepackage{siunitx}
\usepackage{tikz}
	\usetikzlibrary{calc, arrows.meta}
	\usetikzlibrary{decorations.pathreplacing}
	\usetikzlibrary{calligraphy}
	\usetikzlibrary{quotes}
	\usetikzlibrary{shapes.arrows, shapes.geometric}
	\usepackage{listofitems}
\usepackage{pgfplots}
\pgfplotsset{compat=1.18}
	\usepgfplotslibrary{units}
	\usepackage{pgfplotstable}
	
\usepackage{fancyhdr}

\newcommand{\changefont}{
	\color{red}\fontsize{9}{9}\selectfont
}

\let\oldmaketitle\maketitle
\renewcommand{\maketitle}{%
	\oldmaketitle
	\thispagestyle{fancy}
}

\newacronym{acf}{ACF}{autocorrelation function}
\newacronym{bs}{BS}{base station}
\newacronym{cl}{CL}{Cram\'er--Lo\`eve}
\newacronym{cs}{CS}{cyclostationary}
\newacronym{csm}{CSM}{cyclic spectrum matrix}
\newacronym{isd}{IS-Div.}{Itakura--Saito divergence}
\newacronym{kl}{KL}{Karhunen--Lo\`eve}
\newacronym{kld}{KL-Div.}{Kullback--Leibler divergence}
\newacronym{ldct}{LDCT}{Lebesgue's Dominated Convergence Theorem}
\newacronym{psd}{PSD}{power spectral density}
\newacronym{wss}{WSS}{wide-sense stationary}
\newacronym{snr}{SNR}{signal-to-noise ratio}
\newacronym{simo}{SIMO}{single input-multiple output}
\newacronym{mimo}{MIMO}{multiple input-multiple output}
\newacronym{fft}{FFT}{fast Fourier transform}
\newacronym{ml}{ML}{maximum likelihood}
\newacronym{llr}{LLR}{log-likelihood ratio}
\newacronym{rsm}{RSM}{Rihaczek spectrum matrix}
\newacronym{ue}{UE}{user equipment}
\newacronym{sswmod}{SSW-model}{spatially stationary Weichselberger model}
\newacronym{csi}{CSI}{channel state information}
\newacronym{pam}{PAM}{pulse-amplitude modulation}

\newcommand{\cov}[1]{\ensuremath{\mathrm{C}_{#1}}}
\newcommand{\infcov}[1]{\ensuremath{\overset{\infty}{\mathrm{C}}_{#1}}}
\newcommand{\covmat}[1]{\ensuremath{\mathbf{C}_{#1}}}
\newcommand{\cyc}[1]{\ensuremath{\mathrm{S}_{#1}}}
\newcommand{\cycmat}[1]{\ensuremath{\mathbf{S}_{#1}}}
\newcommand{\Ur}{\ensuremath{\mathbf{U}_{\mathrm{R}}}}
\newcommand{\Ut}{\ensuremath{\mathbf{U}_{\mathrm{T}}}}
\newcommand{\BS}{\ensuremath{\mathrm{BS}}}
\newcommand{\snr}{\ensuremath{\mathrm{SNR}}}

\newcommand{\herm}[1]{\ensuremath{#1^{\mathrm{H}}}}
\newcommand{\trans}[1]{\ensuremath{#1^{\mathrm{T}}}}
\newcommand{\id}[1]{\ensuremath{\mathbf{I}_{#1}}}

\newcommand{\power}[1]{\ensuremath{\mathrm{P}_{#1}}}
\newcommand{\euler}{\ensuremath{\mathit{e}}}
\newcommand{\imunit}{\ensuremath{\mathrm{j}}}

\newcommand{\klcyc}[1]{\ensuremath{\mathcal{S}_{#1}}}
\newcommand{\klcycmat}[1]{\ensuremath{\boldsymbol{\mathcal{S}}_{#1}}}

\newcommand{\bsf}[1]{\ensuremath{\boldsymbol{\mathsf{#1}}}}
\newcommand{\vecrand}[1]{\ensuremath{\widetilde{\bsf{#1}}}}
\newcommand{\vecdet}[1]{\ensuremath{\widetilde{\mathbf{#1}}}}
\newcommand{\complex}{\ensuremath{\mathbb{C}}}
\newcommand{\reals}{\ensuremath{\mathbb{R}}}
\newcommand{\integs}{\ensuremath{\mathbb{Z}}}

\newcommand{\nr}{\ensuremath{N_{\mathrm{R}}}}
\newcommand{\nt}{\ensuremath{N_{\mathrm{T}}}}

\newcommand{\frob}{\ensuremath{\mathrm{F}}}
\newcommand{\placeholder}{\ensuremath{\boldsymbol{\cdot}}}

\newcommand{\der}{\ensuremath{\mathrm{d}}}
\newcommand{\Prime}{\ensuremath{\cramped{'\negthinspace}}}

\DeclareMathOperator{\expec}{E}
\DeclareMathOperator*{\argmax}{arg\,max}
\DeclareMathOperator*{\argmin}{arg\,min}
\DeclareMathOperator{\Diag}{Diag}
\DeclareMathOperator{\trace}{Tr}
\DeclareMathOperator{\vect}{vec}
\DeclareMathOperator{\CN}{\mathcal{CN}}
\DeclareMathOperator{\fft}{\mathtt{FFT}}
\DeclareMathOperator{\sinc}{sinc}
\DeclareMathOperator{\rank}{rank}

\newcommand{\ie}{\textit{i.e.} }
\newcommand{\eg}{\textit{e.g.} }

\newtheorem{proposition}{Proposition}

\begin{document}

	\title[]{Low-complexity Detection for Noncoherent Massive MIMO Communications}
	
	\author*[1]{\fnm{Marc} \sur{Vil\`a-Insa}}\email{marc.vila.insa@upc.edu}
	\author[1]{\fnm{Jaume} \sur{Riba}}\email{jaume.riba@upc.edu}
	
	\affil[1]{\orgdiv{Signal Processing and Communications Group (SPCOM)}, \orgname{Universitat Polit\`ecnica de Catalunya (UPC)}, \orgaddress{\city{Barcelona}, \postcode{08034}, \country{Spain}}}
	
	\abstract{%
		This work studies a point-to-point \acrshort{mimo} uplink in which \acrlong{ue} transmits data to a \acrlong{bs} employing a massive array.
		Signal detection is noncoherent and fading is assumed to follow the Weichselberger model.
		By exploiting the spatial stationarity of fading at the \acrlong{bs}, a \acrlong{cs} structure emerges naturally in the space-time representation, which suggests formulating the statistical properties of the received signal in the \acrlong{kl} domain.
		This allows the derivation of a low-complexity receiver that approximates \acrlong{ml} detection even for a moderate array size.
		The spectral analysis of the problem provides valuable insights on the design of space-time codewords.%
	}
	
	
	
	
	\maketitle
	
	\section{Introduction}\label{sec:intro}
	
		The study of events that repeat periodically over time is ubiquitous within many areas of natural and social sciences~\cite{Gardner2006}.
		From climatology~\cite{Schkoda2014} to econometrics~\cite{BroszkiewiczSuwaj2004}, through radio astronomy~\cite{Hellbourg2011} and genomics~\cite{Arora2008}, the efforts to characterize random phenomena with periodic structures have shaped a vast and rich literature over the decades.
		Most of these processes emerge from the interaction between periodic and random circumstances.
		Therefore, although the resulting processes are not periodic themselves, their statistical characterizations are~\cite{Napolitano2013}.
		Of particular interest are those stochastic processes whose second-order moments (\ie correlation structures) are periodic functions of time, which are known as \textit{\acrfull{cs}} or \textit{periodically correlated}~\cite{Gardner2006}.
		This is the kind of statistical periodicity which has been studied more extensively~\cite{Gardner1988,Gardner2006,Napolitano2016}.
		
		As one would expect, man-made signals (\eg the ones used in radar or sonar~\cite{Napolitano2016a}) or those related to human activities (such as rotating machinery~\cite{Antoni2004}) are perfectly captured by the \acrshort{cs} framework.
		Wireless communications are no exception: periodicities occur due to the plethora of signal processing operations involved in modulation, sampling, multiplexing or coding, among many others~\cite{Gardner2006}.
		The pervasiveness of \acrshort{cs} processes in communication systems has motivated the design and development of techniques that exploit their distinct nature~\cite{Gardner1994}.
		They rely on the \acrshort{cs} structure that arises on the oversampled (\ie over the Nyquist rate) communications signals as a means to achieve improved estimation and detection performance.
		As a byproduct, they also offer improved robustness to perform these fundamental tasks in adverse scenarios.
		The advantages come either by avoiding the use of higher order statistics, which become very noisy at low \acrfull{snr} regimes, or by inducing improved richness on the signal space dimensionality.
		For example, in~\cite{Shi2009,Horstmann2018} cyclostationarity is exploited for the estimation of fundamental synchronization parameters, such as timing, symbol rate and carrier frequency, and in~\cite{Tong1995} it is used to perform blind second-order channel estimation without losing phase information.
		In~\cite{Riba2010}, cyclostationarity is proposed as an alternative to moment-based estimation of \acrshort{snr}, whereas works such as~\cite{Pries2018} harness it to perform spectrum sensing in opportunistic communications based on feature detection.

		The work presented herein showcases a novel formulation of a well-known communication problem in which neither the transmitter nor the receiver is aware of instantaneous \acrfull{csi}, referred to as \textit{noncoherent}.
		Such schemes have recently rekindled an interest in the scientific community due to the introduction of potential applications that would benefit from circumventing the acquisition of instantaneous \acrshort{csi}~\cite{Chafii2023}.
		The considered approach exploits periodicities in data for noncoherent detection by harnessing the statistical knowledge of fading (\ie \textit{statistical \acrshort{csi}}) through the \textit{channel hardening} effected by a large array at the receiver~\cite{Jing2016}.
		This proposal differs from classical communications literature in the sense that cyclostationarity does not arise in the time domain.
		Instead, we present a channel model that naturally yields a repeating correlation function, by treating the spatial and temporal dimensions jointly.
		This fading profile, referred to as \textit{\acrlong{sswmod}}, is the result of incorporating the spatial stationarity of the electromagnetic channel~\cite{Marzetta2018} into the well-known correlation model from~\cite{Weichselberger2006}.
		This combination of ideas allows the mathematical machinery developed in~\cite{VilaInsa2024b} for \acrshort{simo} systems to be extended to a \acrshort{mimo} scenario.
	
		The main contributions in this work are the following.
		Beside unveiling the \acrshort{cs} qualities of the fading process, the signal at the receiver is also identified as periodically correlated.
		Leveraging the massive number of antennas to obtain a large sample size, we apply an extension of the Szegö theorem from~\cite{Riba2022,VilaInsa2024} and derive the asymptotic \textit{\acrfull{cl}} and \textit{\acrfull{kl}} expansions of the received signal.
		These representations broaden the frequency domain interpretations from classical signal processing of stationary processes and allow the materialization of clear spectral correlation structures.
		It is in these domains in which the fundamental entities in codeword detection (\eg power spectra and signal subspaces) become conspicuous.
		In particular, we study the asymptotic expression of the \acrfull{ml} detector, from which we obtain insights onto the pairwise codeword detection problem and relate them to other works in the literature.
		As a byproduct of this theoretical analysis, we propose a low-complexity implementation of the \acrshort{ml} receiver that approximates its asymptotic spectral form for a finite number of antenna elements.
		By taking advantage of well-known algorithms like the \acrfull{fft}, the developed technique displays improved computational complexity while retaining near-optimal performance, even for moderately-sized systems.
				
		The paper is structured as follows.
		\autoref{sec:problem} introduces the communication problem and signal detection framework.
		The statistics of fading and how cyclostationarity emerges from them are presented as well.
		\autoref{sec:ml} extends signal detection in the large array regime and exploits statistical periodicities as the number of receiving antennas grows unbounded.
		The limiting spectral expressions provide the foundation for a low-complexity detector with negligible performance penalties, even for moderately sized receivers.
		Finally, \autoref{sec:structure} provides a theoretical overview of this asymptotic formulation through the lens of detection theory.
		Relevant mathematical structures discernible in the spectral domain are linked to physical entities to unveil further insights on noncoherent codeword detection.
		
		The notation used throughout the text is defined next.
		Vectors and matrices: boldface lowercase and uppercase.
		$(r,c)$-component of a matrix: $[\mathbf{A}]_{r,c}$.
		Transpose and conjugate transpose: $\trans{\placeholder}$, $\herm{\placeholder}$.
		Trace: $\trace[\placeholder]$.
		Identity matrix of size $N$: $\id{N}$.
		Column-wise vectorization: $\vect(\placeholder)$.
		Kronecker product: $\otimes$.
		Matrix determinant: $\lvert\mathbf{A}\rvert$.
		Euclidean norm: $\lVert\mathbf{a}\rVert$.
		Frobenius norm: $\lVert\mathbf{X}\rVert_{\frob}$.
		Diagonal matrix constructed from the elements of a set: $\Diag(\placeholder)$.
		Dirac delta distribution: $\delta(\placeholder)$.
		Kronecker delta: $\delta_{a}$.
		Random variables: sans serif font.
		Expectation: $\expec[\placeholder]$.
		Circularly symmetric complex normal vector: $\bsf{a}\sim\CN(\boldsymbol{\mu},\mathbf{C})$.
	
	\section{Problem formulation} \label{sec:problem}
	
		\subsection{Signal model} \label{ssec:model}
	
			Consider a \acrshort{mimo} point-to-point uplink, in which a \acrfull{ue} with $\nt$ antennas communicates wirelessly with a \acrfull{bs} employing a massive array of $\nr$ antennas, such that $\nr>>\nt$ (\textit{asymmetric massive \acrshort{mimo}}).
			The channel displays frequency flat fading with a coherence time of $K\geq\nt$.
			It is represented by a matrix $\bsf{H}\in\complex^{\nt\times\nr}$ that remains constant during $K$ channel uses (\ie block flat fading).
			We assume a power normalization of $\expec[\lVert\bsf{H}\rVert_{\frob}^2]\triangleq\nt\nr$.
			During a coherence block, the transmitter sends a codeword $\bsf{X}\in\complex^{K\times\nt}$ selected equiprobably from a finite alphabet $\mathcal{X}$ of size $M$.
			
			The signal at the \acrshort{bs} (\ie the receiver) is represented by the following complex baseband time-space matrix:
			\begin{equation}
				\bsf{Y}=\bsf{X}\bsf{H}+\bsf{Z}\in\complex^{K\times\nr},
			\end{equation}
			where $\bsf{Z}$ is an additive white Gaussian noise component with i.i.d. entries, such that $[\bsf{Z}]_{r,c}\sim\CN(0,\power{\bsf{Z}})$.
			Fading is assumed correlated Rayleigh with arbitrary covariance.
			To properly characterize its second order statistics, it is convenient to vectorize $\bsf{H}$ column-wise:
			\begin{equation}
				\vecrand{h}\triangleq\vect(\bsf{H})\sim\CN(\mathbf{0}_{\nt\nr},\covmat{\vecrand{h}}), \quad \covmat{\vecrand{h}}\triangleq\expec\bigl[\vecrand{h}\vecrand{h}\herm{\vphantom{\bsf{h}}}\bigr].
			\end{equation}
			Similarly, noise is distributed as $\vecrand{z}\triangleq\vect(\bsf{Z})\sim\CN(\mathbf{0}_{K\nr},\power{\bsf{Z}}\id{K\nr})$.
			Unless stated otherwise, these distributions are known at the receiver but not at the transmitter, which is referred to as \textit{statistical \acrshort{csi} at the receiver}.
			
			The previous vectorization approach can be applied onto the received signal $\bsf{Y}$ as well.
			Moreover, it will prove useful when defining the \acrlong{ml} detector in the sequel.
			Therefore,
			\begin{equation}
				\vecrand{y}\triangleq\vect(\bsf{Y}) = \vecrand{X}\vecrand{h}+\vecrand{z}\in\complex^{K\nr}, \label{eq:vec}
			\end{equation}
			where $\vecrand{X}\triangleq\id{\nr}\otimes\bsf{X}$.
	
		\subsection{Unconditional maximum likelihood detection} \label{ssec:uml}
	
			Given an equiprobable alphabet, the receiver that minimizes the error probability of codeword detection is known to be the \acrfull{ml} detector~\cite[Thm.~21.3.3]{Lapidoth2017}.
			It is derived from the distribution of the received signal conditioned to a transmitted codeword $\mathbf{X}_{i}\in\mathcal{X}$, and a channel realization $\mathbf{H}$:
			\begin{equation}
				\mathrm{f}_{\bsf{Y}\vert\mathbf{X}_{i},\mathbf{H}}(\vecdet{y})=\frac{1}{(\pi\power{\bsf{Z}})^{K\nr}}\exp\bigl(-\tfrac{1}{\power{\bsf{Z}}}\lVert\vecdet{y}-\vecdet{X}_{i}\vecdet{h}\rVert^2\bigr).\label{eq:conditional}
			\end{equation}		
			In a noncoherent communication setting, the channel realization is unknown at both ends of the link and treated as random (\ie \textit{unconditional model}~\cite{Stoica1990}).
			It is removed from conditioning by marginalizing~\eqref{eq:conditional}.
			The resulting distribution of $\bsf{Y}\vert\mathbf{X}_i$ becomes
			\begin{equation}
				\mathrm{f}_{\bsf{Y}\vert\mathbf{X}_{i}}(\vecdet{y}) = \expec_{\bsf{H}}\bigl[\mathrm{f}_{\bsf{Y}\vert\mathbf{X}_{i},\mathbf{H}}(\vecdet{y})\bigr] = \frac{1}{\pi^{K\nr}\lvert\covmat{i}\rvert}\exp\bigl(-\vecdet{y}\herm{\vphantom{\mathbf{y}}}\covmat{i}^{-1}\vecdet{y}\bigr), \label{eq:likelihood}
			\end{equation}
			where
			\begin{equation}
				\covmat{i} \triangleq \vecdet{X}_{i}\covmat{\vecrand{h}}\vecdet{X}\herm{\vphantom{\mathbf{X}}}_{i} + \power{\bsf{Z}}\id{K\nr}
			\end{equation}
			is the covariance matrix of the received signal.
			With the likelihood function~\eqref{eq:likelihood}, we finally define the normalized \acrshort{ml} detector as
			\begin{equation}
				\begin{aligned}
					\widehat{\mathbf{X}}_{\textsc{ml}} \triangleq \argmax_{\mathbf{X}_{j}\in\mathcal{X}} \mathrm{f}_{\bsf{Y}\vert\mathbf{X}_{j}}(\vecdet{y}) &=  \argmin_{\mathbf{X}_{j}\in\mathcal{X}} -\tfrac{1}{K\nr}\ln\mathrm{f}_{\bsf{Y}\vert\mathbf{X}_{j}}(\vecdet{y}) \triangleq \argmin_{\mathbf{X}_{j}\in\mathcal{X}} \mathcal{L}_j(\vecdet{y})\\
					&= \argmin_{\mathbf{X}_{j}\in\mathcal{X}} \tfrac{1}{K\nr}\vecdet{y}\herm{\vphantom{\mathbf{y}}}\covmat{j}^{-1}\vecdet{y} + 	\tfrac{1}{K\nr}\ln\lvert\covmat{j}\rvert.
				\end{aligned} \label{eq:ml}
			\end{equation}
			It is composed of a data-dependent quadratic term and a logarithmic term:
			\begin{equation}
				\mathrm{Q}_{j}(\vecdet{y})\triangleq\tfrac{1}{K\nr}\vecdet{y}\herm{\vphantom{\mathbf{y}}}\covmat{j}^{-1}\herm{\vecdet{y}} \quad , \quad \ell_{j}\triangleq\tfrac{1}{K\nr}\ln\lvert\covmat{j}\rvert, \label{eq:ml_terms}
			\end{equation}
			such that $\mathcal{L}_{j}(\vecdet{y}) = \mathrm{Q}_{j}(\vecdet{y}) + \ell_{j}$.
			
			From a theoretical perspective,~\eqref{eq:ml} entails some analysis issues.
			If $\covmat{\vecrand{h}}$ is an unstructured, arbitrary covariance matrix, the set of \acrshort{ml} metrics $\{\mathcal{L}_{j}(\vecdet{y})\}$ does not provide any recognizable insights onto the detection problem nor onto the design of $\mathcal{X}$ based on $\{\mathrm{Q}_{j}(\vecdet{y})\}$ and $\{\ell_{j}\}$.
			Therefore, the next section is devoted to defining a more specific correlation model for $\bsf{H}$.
			Its statistical properties, paired with the use of a large array at the \acrshort{bs}, will allow us to represent the received signal through its \acrlong{kl} and \acrlong{cl} expansions in \autoref{sec:ml}.
			It is in such spectral domains in which the detection-relevant structures emerge.
			They are properly discussed in \autoref{sec:structure}, where we relate them to the idea of \textit{singular detection}~\cite{VilaInsa2025}.
			As a byproduct of this exploration, a low-complexity implementation of~\eqref{eq:ml} will arise and be assessed in \autoref{ssec:low_complex}.
	
		\subsection{Channel model and spatial stationarity} \label{ssec:channel_model}
	
			The model that will be assumed for the channel throughout the rest of this work is known as \textit{Weichselberger model}~\cite{Weichselberger2006}, which is a particularization of correlated Rayleigh fading:
			\begin{equation}
				\bsf{H}\triangleq\Ut\mathring{\bsf{H}}\herm{\Ur},\label{eq:weichselberger}
			\end{equation}
			where $\Ut\in\complex^{\nt\times\nt}$ and $\Ur\in\complex^{\nr\times\nr}$ are fixed unitary bases that represent the \acrshort{ue} and \acrshort{bs} array geometries, respectively.
			The only stochastic parameters in this model are contained in $\mathring{\bsf{H}}\in\complex^{\nt\times\nr}$, which are independent and distributed as $[\mathring{\bsf{H}}]_{r,c}\sim\CN(0,\gamma_{r,c})$.
			It generalizes various classical \textit{nongeometrical stochastic models} of fading~\cite{Feng2022}, such as i.i.d. Rayleig, keyhole and virtual channel (\textit{beam-domain})~\cite[Sec.~3.6.1]{Heath2018}.
			By relaxing the structural constraints of the \textit{Kronecker model} and allowing correlations between transmitter and receiver, it solves some of its shortcomings without fully sacrificing signal structure.
			
			To simplify subsequent analyses, we will make further assumptions on the statistical properties of $\bsf{H}$.
			It is common in the literature to consider channel fading behaves as a \textit{spatially stationary} process, \ie its second-order moments are invariant to spatial shifts~\cite{Sanguinetti2020}.
			This idea can be effectively harnessed when using a massive array (such as in the \acrshort{bs} of the studied system), in the same manner as how it is done in the time domain when a large data set is available~\cite{Gray2005}, thus giving rise to the classical frequency analysis toolkit grounded on the core property that Fourier coefficients at different frequencies become uncorrelated~\cite{Vaidyanathan2008}.
			Therefore, we only need stationarity to manifest at the \acrshort{bs} side of the wireless link, which translates into the rows of $\bsf{H}$ being jointly \acrshort{wss}:
			\begin{equation}
				\expec\bigl[[\bsf{H}]_{r,c}[\bsf{H}]_{r\Prime,c\Prime}^*\bigr] \equiv \expec\bigl[[\bsf{H}]_{r,c+l}[\bsf{H}]_{r\Prime,c\Prime+l}^*\bigr], \quad \forall l\in\integs. \label{eq:joint_wss}
			\end{equation}
			
			To expand on how the previous notion can be exploited, we define the correlation matrix
			\begin{equation}
				\mathbf{R}_{\BS} \triangleq 		\expec[\herm{\bsf{H}}\bsf{H}] = \Ur\expec[\mathring{\bsf{H}}\herm{\vphantom{\bsf{H}}}\mathring{\bsf{H}}]\herm{\Ur} = \Ur\Diag\bigl(\textstyle\bigl\{\sum_{r=0}^{\nt-1}\gamma_{r,c}\bigr\}_{c}\bigr)\herm{\Ur}\in\complex^{\nr\times\nr}, \label{eq:bs_cov}
			\end{equation}
			in which we have used the Weichselberger model from~\eqref{eq:weichselberger}.
			The channel being spatially stationary at the \acrshort{bs} implies that $[\mathbf{R}_{\BS}]_{r,c} \equiv [\mathbf{R}_{\BS}]_{r+l,c+l}$ for any $l$, by~\eqref{eq:joint_wss}.
			Therefore, $\mathbf{R}_{\BS}$ is constant along its diagonals which is known as a \textit{Toeplitz matrix} (see \autoref{subfig:toeplitz}).
			This type of matrices have the special property of asymptotically diagonalizing onto the Fourier basis $[\mathbf{F}_{\nr}]_{r,c}\triangleq \euler^{\imunit2\pi rc/\nr}/\sqrt{\nr}$ as $\nr\to\infty$~\cite{Sanguinetti2020}.
			The interest behind this structure is that it can be exploited with well-known precoding implementations, such as the phase-shifter architectures used for beamspace techniques~\cite{Zhang2020}.
			
			Under the premise of using a massive array in the \acrshort{bs}, and to simplify upcoming mathematical derivations, we will assume that $\Ur\triangleq\herm{\mathbf{F}_{\nr}}$, even for finite $\nr$.
			As the number of receiving antennas grows without bound, the previously described Fourier basis will asymptotically become the true eigenbasis of $\bsf{H}$.
			With this,~\eqref{eq:weichselberger} becomes the \textit{\acrfull{sswmod}}:
			\begin{equation}
				\bsf{H}\triangleq\Ut\mathring{\bsf{H}}\mathbf{F}_{\nr}.\label{eq:ssw}
			\end{equation}
			Returning to the formulation from \autoref{ssec:model}, we may want to express it in vector form:
			\begin{equation}
				\vecrand{h}=(\mathbf{F}_{\nr}\otimes\Ut)\vect(\mathring{\bsf{H}}).\label{eq:vec_ssw}
			\end{equation}
			Its covariance matrix is
			\begin{equation}
				\covmat{\vecrand{h}} = (\mathbf{F}_{\nr}\otimes\Ut) \mathbf{\Gamma} (\herm{\mathbf{F}_{\nr}}\otimes\herm{\Ut}),
			\end{equation}
			where $\mathbf{\Gamma} \triangleq \expec[\vect(\mathring{\bsf{H}})\vect\herm{(\mathring{\bsf{H}})}] \in\reals^{\nt\nr\times\nt\nr}$ is a diagonal matrix which contains its eigenvalues.
			By~\eqref{eq:weichselberger}, we know they are nonnegative and $[\mathbf{\Gamma}]_{r+c\nt,r+c\nt} = \gamma_{r,c}$.
			
			\begin{figure}[t]
				\centering
				\begin{tikzpicture}
					\pgfmathsetmacro{\arrowlen}{3}
					\pgfmathsetmacro{\sqlen}{0.5}
					\readlist*\colorlist{Firebrick1,SpringGreen3,RoyalBlue1}
					\tikzset{	
						mysquare/.pic={
							\path[pic actions] +($(\sqlen/10,\sqlen/10)$) rectangle +($(9*\sqlen/10,9*\sqlen/10)$);
							\draw[thick] +(0,0) rectangle +(\sqlen,\sqlen);
						}
					}
					\node[right] (H) at (0,0) {$\bsf{H}=$};
					\foreach \ntx in {1,2,3}
					\foreach \nrx in {0,1,2}{	
						\itemtomacro\colorlist[\ntx]\thiscolor
						\pgfmathparse{int(100-(\nrx*25))}
						\draw ($(H.east)+(\nrx*\sqlen,\ntx*\sqlen - 2.5*\sqlen)$) pic[fill=\thiscolor!\pgfmathresult!white]{mysquare};
					}
					\node[
					draw,
					thick,
					rectangle,
					minimum width=1.5cm*\sqlen,
					minimum height=3cm*\sqlen
					] (dots) at ($(H.east)+(3.75*\sqlen,0)$) {$\cdots$};
					\foreach \ntx in {1,2,3}{
						\itemtomacro\colorlist[\ntx]\thiscolor
						\draw ($(dots.center)+(0.75*\sqlen,\ntx*\sqlen - 2.5*\sqlen)$) pic[fill=\thiscolor!25!white]{mysquare};
					}
					\coordinate (matrix) at ($(H.east)+(5.5*\sqlen,0)$);
					\path[
					ultra thick,
					arrows={<->},
					color=gray
					] ($(matrix)+(-5.5*\sqlen,2*\sqlen)$) edge["WSS"] +($(5.5*\sqlen,0)$);
					\node[rotate=90] (ni) at ($(matrix)-(2.75*\sqlen,2*\sqlen)$) {$\ni$};
					\node[anchor=north,rectangle,draw,below] at (ni.west) {$\complex^{\nt\times\nr}$};
					\coordinate (vector) at ($(matrix)+(\arrowlen,0)$);
					\path[	
					ultra thick,
					arrows={Circle[open]-Stealth[]}
					] (matrix) edge["vectorization"] (vector);
					\foreach \ntx in {1,2,3}{ 
						\itemtomacro\colorlist[\ntx]\thiscolor
						\draw ($(vector)+(0,0.5*\sqlen+\ntx*\sqlen)$) pic[fill=\thiscolor]{mysquare};
						\draw ($(vector)+(0,-2.5*\sqlen+\ntx*\sqlen)$) pic[fill=\thiscolor!75!white]{mysquare};
					}
					\node[
					draw,
					thick,
					rectangle,
					minimum width=1cm*\sqlen,
					minimum height=3cm*\sqlen
					] at ($(vector)+(0.5*\sqlen,-3*\sqlen)$) {$\vdots$};
					\node[right] at ($(vector)+(\sqlen,0)$) {$=\vecrand{h}\in\complex^{\nt\nr}$};
					\draw[
					decorate,
					decoration={calligraphic brace, amplitude=\sqlen*0.5cm},
					thick
					] ($(vector)+(\sqlen,4.5*\sqlen)$) -- +($(0,-3*\sqlen)$)  node[midway,anchor=west,right=\sqlen*0.5cm]{$\nt=3$};
				\end{tikzpicture}
				\caption{Graphical representation of column-wise vectorization of $\bsf{H}$ under the \acrshort{sswmod}.}
				\label{fig:vec}
			\end{figure}
			
			The following result sheds light onto the statistical structure of $\vecrand{h}$.
			A graphical intuition behind it can be found in \autoref{fig:vec}.
			\begin{proposition}\label{prop:H}
				Under the \acrshort{sswmod}, fading behaves as a \acrshort{cs} process of period $\nt$.
			\end{proposition}
			\begin{proof}
				Recall~\eqref{eq:joint_wss}.
				Expressing this relationship in terms of the vectorized form $\vecrand{h}$ yields
				\begin{subequations}
					\begin{align}
						\expec\bigl[[\vecrand{h}]_{r+c\nt}[\vecrand{h}]_{r\Prime+c\Prime\nt}^*\bigr] &\equiv \expec\bigl[[\vecrand{h}]_{r+(c+l)\nt}[\vecrand{h}]_{r\Prime+(c\Prime+l)\nt}^*\bigr] \\
						[\covmat{\vecrand{h}}]_{r+c\nt,r\Prime+c\Prime\nt} &= [\covmat{\vecrand{h}}]_{(r+c\nt)+l\nt,(r\Prime+c\Prime\nt)+l\nt}.
					\end{align}
				\end{subequations}
				By setting $p\triangleq r+c\nt$ and $q\triangleq r\Prime+c\Prime\nt$, we have that
				\begin{equation}
					[\covmat{\vecrand{h}}]_{p,q}=[\covmat{\vecrand{h}}]_{p+l\nt,q+l\nt},
				\end{equation}
				which is a defining property of $\nt$-periodic \acrshort{cs} processes.
			\end{proof}
			\acrshort{cs} processes have been widely studied in the literature~\cite{Gardner2006}, since they display second order statistics that are invariant to shifts of length equal to their period.
			For $\vecrand{h}$, this translates into a covariance matrix that is periodic along its diagonals with period $\nt$, known as $\nt$\textit{-Toeplitz}.
			These matrices are also referred to as \textit{block-Toeplitz}, since they have constant blocks of size $\nt$ diagonally (see \autoref{subfig:block_toeplitz}).
			They play an important role in signal processing, and have been thoroughly studied in the literature~\cite{Ramirez2015,Riba2022}.
			
			\begin{figure}[t]
				\centering
				\begin{subfigure}{0.49\linewidth}
					\centering
					\begin{tikzpicture}
						\pgfmathsetmacro{\sqlen}{.5}
						\pgfmathsetmacro{\xmax}{6}
						\pgfmathsetmacro{\ymax}{6}
						\readlist*\colorlist{MediumPurple2,Tan2,Gold2,SpringGreen2,RoyalBlue2,Firebrick2,RoyalBlue2,SpringGreen2,Gold2,Tan2,MediumPurple2}
						\tikzset{	
							mysquare/.pic={
								\path[pic actions] +($(0.1*\sqlen,0.1*\sqlen)$) rectangle +($(0.9*\sqlen,0.9*\sqlen)$);
								\draw[thick] +(0,0) rectangle +(\sqlen,\sqlen);
							}
						}
						\foreach \x in {1,...,\xmax}{
							\foreach \y in {1,...,\ymax}{
								\pgfmathparse{int(\x+\y-1)}
								\itemtomacro\colorlist[\pgfmathresult]\thiscolor
								\draw ($(\x*\sqlen-\sqlen,\y*\sqlen-\sqlen)$) pic[fill=\thiscolor]{mysquare};
							}
						}
						\foreach \x in {1,...,\xmax}{
							\draw[draw=black] 	($(\x*\sqlen,\ymax*\sqlen)$) -- ($(\xmax*\sqlen,\x*\sqlen)$);
							\draw[draw=black] 	($(0,\ymax*\sqlen-\x*\sqlen)$) -- ($(\xmax*\sqlen-\x*\sqlen,0)$);
						}
						\draw[draw=black] 	($(0,\ymax*\sqlen)$) -- ($(\xmax*\sqlen,0)$);
					\end{tikzpicture}
					\subcaption{Toeplitz matrix.}
					\label{subfig:toeplitz}
				\end{subfigure}
				\begin{subfigure}{0.49\linewidth}
					\centering
					\begin{tikzpicture}
						\pgfmathsetmacro{\sqlen}{.5}
						\pgfmathsetmacro{\xmax}{6}
						\tikzset{
							mysquare/.pic={
								\path[pic actions] +($(0.1*\sqlen,0.1*\sqlen)$) rectangle +($(0.9*\sqlen,0.9*\sqlen)$);
							}
						}
						\tikzset{
							blockblue/.pic={
								\draw (0,0)	pic[fill=RoyalBlue4] 	{mysquare};
								\draw (0,\sqlen) pic[fill=Firebrick4] 	{mysquare};
								\draw (\sqlen,0) pic[fill=Firebrick1] 	{mysquare};
								\draw (\sqlen,\sqlen) pic[fill=RoyalBlue4] 	{mysquare};
								\draw[draw=black, thick]	(0,0) rectangle ($(2*\sqlen,2*\sqlen)$);
							}
						}
						\tikzset{
							blockred/.pic={
								\draw (0,0)	pic[fill=MediumPurple4] {mysquare};
								\draw (0,\sqlen) pic[fill=Tan4] {mysquare};
								\draw (\sqlen,0) pic[fill=Tan1] {mysquare};
								\draw (\sqlen,\sqlen) pic[fill=Gold1] {mysquare};
								\draw[draw=black, thick]	(0,0) rectangle ($(2*\sqlen,2*\sqlen)$);
							}
						}
						\tikzset{
							blockredT/.pic={
								\draw (0,0)	pic[fill=Gold1] {mysquare};
								\draw (0,\sqlen) pic[fill=Tan4] {mysquare};
								\draw (\sqlen,0) pic[fill=Tan1] {mysquare};
								\draw (\sqlen,\sqlen) pic[fill=MediumPurple4] {mysquare};
								\draw[draw=black, thick]	(0,0) rectangle ($(2*\sqlen,2*\sqlen)$);
							}
						}
						\tikzset{
							blockgreen/.pic={
								\draw (0,0)	pic[fill=Gold4] {mysquare};
								\draw (0,\sqlen) pic[fill=SpringGreen4] {mysquare};
								\draw (\sqlen,0) pic[fill=SpringGreen1] {mysquare};
								\draw (\sqlen,\sqlen) pic[fill=RoyalBlue1] {mysquare};
								\draw[draw=black, thick]	(0,0) rectangle ($(2*\sqlen,2*\sqlen)$);
							}
						}
						\tikzset{
							blockgreenT/.pic={
								\draw (0,0)	pic[fill=RoyalBlue1] {mysquare};
								\draw (0,\sqlen) pic[fill=SpringGreen4] {mysquare};
								\draw (\sqlen,0) pic[fill=SpringGreen1] {mysquare};
								\draw (\sqlen,\sqlen) pic[fill=Gold4] {mysquare};
								\draw[draw=black, thick]	(0,0) rectangle ($(2*\sqlen,2*\sqlen)$);
							}
						}
						\draw (0,0) pic {blockred};
						\draw ($(4*\sqlen,4*\sqlen)$ )pic {blockredT};
						\draw ($(2*\sqlen,0)$) pic {blockgreen};
						\draw ($(0,2*\sqlen)$) pic {blockgreen};
						\draw ($(2*\sqlen,4*\sqlen)$) pic {blockgreenT};
						\draw ($(4*\sqlen,2*\sqlen)$) pic {blockgreenT};
						\draw ($(0,4*\sqlen)$) pic {blockblue};
						\draw ($(2*\sqlen,2*\sqlen)$) pic {blockblue};
						\draw ($(4*\sqlen,0)$) pic {blockblue};
						\foreach \x in {1,...,\xmax}{
							\draw[draw=black, dashed] 	($(\x*\sqlen,\xmax*\sqlen)$) -- ($(\xmax*\sqlen,\x*\sqlen)$);
							\draw[draw=black, dashed] 	($(0,\xmax*\sqlen-\x*\sqlen)$) -- ($(\xmax*\sqlen-\x*\sqlen,0)$);
						}
						\draw[draw=black, dashed] 	($(0,\xmax*\sqlen)$) -- ($(\xmax*\sqlen,0)$);
						\draw[decorate, decoration={calligraphic brace, amplitude=\sqlen*0.5cm}, thick] ($(\xmax*\sqlen,\xmax*\sqlen)$) -- ++($(0,-2*\sqlen)$)
						node[midway,anchor=west,right=\sqlen*0.5cm]{$\nt=2$};
					\end{tikzpicture}
					\subcaption{$\nt$-Toeplitz matrix.}
					\label{subfig:block_toeplitz}
				\end{subfigure}
				\caption{Graphical representation of the covariance matrix of a \acrshort{wss} process and a \acrshort{cs} process of period 2.}
				\label{fig:toeplitz}
			\end{figure}
			
		\subsubsection{Fading statistics in the large array regime}
			
			To conclude this section, we will present the fading correlation in the large array regime, \ie as $\nr\to\infty$.
			We define the \acrfull{acf} of $\{\mathsf{h}(n)\triangleq[\vecrand{h}]_n\}$, for $n\in\{-\frac{\nt\nr}{2},\dots,\frac{\nt\nr}{2}-1\}$ as\footnote{We assume $\nr\in2\mathbb{N}$ without loss of generality.}
			\begin{equation}
				\cov{\bsf{H}}(n,m) \triangleq\expec[\mathsf{h}(n+m)\mathsf{h}(n)]=[\covmat{\vecrand{h}}]_{n+m,n}.
			\end{equation}
			Under model~\eqref{eq:ssw}, it takes the form
			\begin{equation}
				\cov{\bsf{H}}(n,m)= \smashoperator{\sum_{p,q=-\frac{\nr}{2}}^{\frac{\nr}{2}-1}} [\Ut\covmat{\bsf{H}}(p-q)\herm{\Ut}]_{n+m-p\nt,n-q\nt}, \label{eq:channel_acf}
			\end{equation}
			where $\covmat{\bsf{H}}(m)\in\complex^{\nt\times\nt}$ is a diagonal matrix whose $k$th entry is the \acrshort{acf} of $\{\mathring{\mathsf{h}}_k(n)\triangleq[\mathring{\bsf{H}}\mathbf{F}_{\nr}]_{k,n}\}$, which is a \acrfull{wss} process.
			Indeed, its correlation is shift-invariant:
			\begin{equation}
				\expec\bigl[\mathring{\mathsf{h}}_{k}(n+m)\mathring{\mathsf{h}}_{k\Prime}^*(n)\bigr] = \delta_{k-k\Prime}\smashoperator{\sum_{l=0}^{\nr-1}}\tfrac{1}{\nr}\gamma_{k,l}\euler^{\imunit2\pi m\frac{l}{\nr}} \triangleq [\covmat{\bsf{H}}(m)]_{k,k\Prime}.\label{eq:channel_acf_entry}
			\end{equation}
			Finally, the next result states the limit of~\eqref{eq:channel_acf} in the large array regime.
			
			\begin{proposition} \label{prop:channel}
				As $\nr$ grows without bound, the fading \acrshort{acf}~\eqref{eq:channel_acf} approaches
				\begin{equation}
					\lim_{\nr\to\infty} \cov{\bsf{H}}(n,m) = \smashoperator[l]{\sum_{p,q\in\integs}} \int_{0}^{\mathrlap{1}}[\Ut\cycmat{\bsf{H}}(\lambda)\herm{\Ut}]_{n+m-p\nt,n-q\nt}\euler^{\imunit2\pi(p-q)\lambda}\der\lambda, \label{eq:channel_acf_inf}
				\end{equation}
				where
				\begin{equation}
					\cycmat{\bsf{H}}(\lambda)\triangleq \smashoperator{\sum_{m\in\integs}}\covmat{\bsf{H}}(m)\euler^{-\imunit2\pi m\lambda} \label{eq:channel_psd}
				\end{equation}
				is a diagonal matrix whose $k$th entry, $\cyc{\bsf{H}}^{(k)}(\lambda)$, contains the \acrfull{psd} of $\{\mathring{\mathsf{h}}_k(n)\}$.
			\end{proposition}
			\begin{proof}
				See \autoref{app:proof_channel_acf}.
			\end{proof}
			
	\section{\texorpdfstring{\acrshort{ml}}{ML} detection in the large array regime} \label{sec:ml}
		
		This section deals with the received signal at the \acrshort{bs} and how to exploit its massive array to obtain various insights on the detection of $\{\mathbf{X}_{i}\in\mathcal{X}\}$ and their design.
		
		The received signal from~\eqref{eq:vec}, $\{\mathsf{y}_{i}(n)\triangleq[\vecrand{y}\vert\mathbf{X}_{i}]_{n}\}$, for $n\in\{-\frac{K\nr}{2},\dots,\frac{K\nr}{2}-1\}$, can be expressed as
		\begin{equation}
			\mathsf{y}_{i}(n) = \smashoperator{\sum_{l=-\frac{\nr}{2}}^{\frac{\nr}{2}-1}}[\mathbf{X}_{i}\bsf{H}]_{n-lK,l} + [\vecrand{z}]_n. \label{eq:y_i}
		\end{equation}
		Its \acrshort{acf} is defined as
		\begin{equation}
			\cov{i}(n,m) \triangleq \expec[\mathsf{y}_{i}(n+m)\mathsf{y}_{i}^*(n)].\label{eq:y_acf}
		\end{equation}
		\begin{proposition} \label{prop:cs}
			Under the \acrshort{sswmod}, the received signal at the \acrshort{bs} is a \acrshort{cs} process of period $K$, \ie $\cov{i}(n,m)\equiv\cov{i}(n+rK,m)$.
		\end{proposition}
		\begin{proof}
			See \autoref{app:proof_cs}.
		\end{proof}
		\begin{figure}[b]
			\centering
			\begin{tikzpicture}
				\pgfmathsetmacro{\arrowlen}{3}
				\pgfmathsetmacro{\sqlen}{0.5}
				\coordinate (matrix) at ($(7*\sqlen,5*\sqlen)$);
				\coordinate (vector) at ($(matrix) + (0.5*\sqlen,0)$);
				\coordinate (result) at ($(vector) + (7*\sqlen,0)$);
				\tikzset{	
					mysquare/.pic={
						\draw[thick, fill=white] +(0,0) rectangle +(\sqlen,\sqlen);
						\path[pic actions] +($(\sqlen/10,\sqlen/10)$) rectangle +($(9*\sqlen/10,9*\sqlen/10)$);
					}
				}
				\fill[color=Firebrick1!40!white] ($(7*\sqlen,10*\sqlen)$) rectangle (0,0);
				\fill[color=Firebrick1!20!white] (0,0) -- ++($(-3.5*\sqlen,\sqlen)$) -- ++($(0,2*\sqlen)$) -- ++($(3.5*\sqlen,7*\sqlen)$) -- cycle;
				\node[
					fill=Firebrick1!40!white,
					shape=rectangle,
					minimum width=1.5cm*\sqlen,
					minimum height=2cm*\sqlen] at ($(-2.75*\sqlen,2*\sqlen)$) {$\vecdet{X}_i$};
				\fill[color=gray!40!white] ($(vector) - (0,5*\sqlen)$) rectangle +($(2*\sqlen,10*\sqlen)$);
				\fill[color=gray!20!white] ($(vector) + (2*\sqlen,-5*\sqlen)$) -- ++($(3*\sqlen,\sqlen)$) -- ++($(0,2*\sqlen)$) -- ++($(-3*\sqlen,7*\sqlen)$) -- cycle;
				\node[
					fill=gray!40!white,
					shape=rectangle,
					minimum width=1cm*\sqlen,
					minimum height=2cm*\sqlen] at ($(vector) + (4.5*\sqlen,-3*\sqlen)$) {$\vecrand{h}$};
				\fill[color=SkyBlue1!40!white] ($(result) - (0,5*\sqlen)$) rectangle +($(2*\sqlen,10*\sqlen)$);
				\fill[fill=SkyBlue1!20!white] ($(result) + (2*\sqlen,-5*\sqlen)$) -- ++($(4*\sqlen,\sqlen)$) -- ++($(0,2*\sqlen)$) -- ++($(-4*\sqlen,7*\sqlen)$);
				\node[
					fill=SkyBlue1!40!white,
					shape=rectangle,
					minimum width=2cm*\sqlen,
					minimum height=2cm*\sqlen] at ($(result) + (5*\sqlen,-3*\sqlen)$) {$\vecdet{X}_i\vecrand{h}$};
				\foreach \bl in {0,1}
				\foreach \k in {1,2,3}
				\foreach \ntx in {1,2}{	
					\draw ($(\sqlen*\ntx-0.5*\sqlen+\sqlen*2*\bl,\sqlen*9.5-\sqlen*\k-\sqlen*3*\bl)$) pic[fill=Firebrick1]{mysquare};
				}
				\foreach \ntx in {1,2}{ 
					\pgfmathparse{int(100-30*\ntx+30)}
					\draw ($(vector)+(0.5*\sqlen,4.5*\sqlen-\ntx*\sqlen)$) pic[fill=RoyalBlue1!\pgfmathresult!white]{mysquare};
					\draw ($(vector)+(0.5*\sqlen,2.5*\sqlen-\ntx*\sqlen)$) pic[fill=SpringGreen3!\pgfmathresult!white]{mysquare};
				}
				\foreach \ntx in {1,2,3}{ 
					\pgfmathparse{int(100-30*\ntx+30)}
					\draw ($(result)+(0.5*\sqlen,4.5*\sqlen-\ntx*\sqlen)$) pic[fill=MediumPurple1!\pgfmathresult!white]{mysquare};
					\draw ($(result)+(0.5*\sqlen,1.5*\sqlen-\ntx*\sqlen)$) pic[fill=Gold1!\pgfmathresult!white]{mysquare};
				}
				\node[
					draw,
					thick,
					rectangle,
					minimum width=1cm*\sqlen,
					minimum height=3cm*\sqlen,
					fill=white
				] at ($(result)+(\sqlen,-3*\sqlen)$) {\large$\vdots$};
				\node[
					fill=white,
					draw,
					thick,
					minimum width=2cm*\sqlen,
					minimum height=3cm*\sqlen
				] at ($(5.5*\sqlen,2*\sqlen)$) {};
				\node at ($(5.5*\sqlen,2.5*\sqlen)$) {$\vdots$};
				\node at ($(5.5*\sqlen,1*\sqlen)$) {$\times\nr$};
				\node at ($(matrix) - (5.5*\sqlen,0)$) {$\mathbf{X}\triangleq$};
				\node[
					draw,
					thick,
					rectangle,
					minimum width=1cm*\sqlen,
					minimum height=5cm*\sqlen,
					fill=white
				] at ($(vector)+(\sqlen,-2*\sqlen)$) {\large$\vdots$};
				\node[
					single arrow,
					draw,
					thick,
					minimum height=2cm*\sqlen
				] (arrow) at ($(result) + (-1.5*\sqlen,0)$) {};
				\draw[
					decorate,
					decoration={calligraphic brace, amplitude=\sqlen*0.3cm},
					thick
				] ($(2.5*\sqlen,9.5*\sqlen)$) -- ++($(0,-3*\sqlen)$) node[midway,anchor=west,right=\sqlen*0.3cm]{$K=3$};
				\draw[
					decorate,
					decoration={calligraphic brace, amplitude=\sqlen*0.3cm},
					thick
				] ($(4.5*\sqlen,3.5*\sqlen)$) -- ++($(-2*\sqlen,0)$)
				node[
					pos=0.65,
					below=0.3cm*\sqlen
				] {$\nt=2$};
				\draw[draw=black] 	;
				\draw[
					decorate,
					decoration={calligraphic brace, amplitude=\sqlen*0.3cm},
					thick
				] ($(vector)+(1.5*\sqlen,4.5*\sqlen)$) -- ++(0,-2*\sqlen) node[midway,anchor=west,right=\sqlen*0.3cm]{$\nt$};
				\draw[
					decorate,
					decoration={calligraphic brace, amplitude=\sqlen*0.3cm},
					thick
				] ($(result)+(1.5*\sqlen,4.5*\sqlen)$) -- ++(0,-3*\sqlen) node[midway,anchor=west,right=\sqlen*0.3cm]{$K$};
			\end{tikzpicture}
			\caption{Visual illustration of the change of period between $\vecrand{h}$ and $\vecdet{X}_i\vecrand{h}$.}
			\label{fig:period}
		\end{figure}
		
		This result can be easily understood by considering the transformation between $\vecrand{h}$ and $\vecrand{y}_i$, which has been illustrated in \autoref{fig:period}.
		By multiplying $\vecrand{h}$ by $\vecdet{X}_i$, we transform a \acrshort{cs} process into one with period $K$.
		Then, by adding the noise term $\vecrand{z}$, which is \acrshort{wss}, the \acrshort{cs} properties of $\vecdet{X}_i\vecrand{h}$ are preserved and $\vecrand{y}_i$ is a $K$-periodic \acrshort{cs} process.
		Its $K$-Toeplitz covariance matrix is
		\begin{equation}
			\covmat{i}\triangleq\begin{bmatrix}
				\covmat{i}(0) & \covmat{i}(-1) 	& \cdots & \covmat{i}(-\nr+1) \\
				\covmat{i}(1) & \covmat{i}(0) & \cdots & \covmat{i}(-\nr+2) \\
				\vdots	& \vdots & \ddots & \vdots \\
				\covmat{i}(\nr-1) & \covmat{i}(\nr-2) & \cdots & \covmat{i}(0)
			\end{bmatrix}, \label{eq:blocks}
		\end{equation}
		with blocks $\covmat{i}(n)\in\complex^{K\times K}$ for $n\in\{-\nr+1,\dots,\nr-1\}$.
		Note that
		\begin{equation}
			\covmat{i}(n)\equiv\herm{\covmat{i}}(-n) \label{eq:herm}.
		\end{equation}
		
		Under the \acrshort{sswmod}, $\cov{i}(n,m)$ from~\eqref{eq:y_acf} (see~\eqref{eq:Ci}) may be developed as
		\begin{equation}
			\cov{i}(n,m) = \power{\bsf{Z}}\delta_m + \smashoperator{\sum_{l,l\Prime=-\frac{\nr}{2}}^{\frac{\nr}{2}-1}}[\mathbf{X}_{i}\Ut\covmat{\bsf{H}}(l-l\Prime)\herm{\Ut}\herm{\mathbf{X}_{i}}]_{n+m-lK,n-l\Prime K},
		\end{equation}
		by applying~\eqref{eq:channel_acf_entry}.
		The obtained \acrshort{acf} can be asymptotically expressed as
		\begin{equation}
			\begin{aligned}
				\infcov{i}(n,m) &\triangleq \lim_{\nr\to\infty} \cov{i}(n,m) \\
				&= \power{\bsf{Z}}\delta_m + \smashoperator[l]{\sum_{l,l\Prime\in\integs}} \int_{0}^{\mathrlap{1}} [\mathbf{X}_{i}\Ut\cycmat{\bsf{H}}(\lambda)\herm{\Ut}\herm{\mathbf{X}_{i}}]_{n+m-lK,n-l\Prime K}\euler^{\imunit2\pi(l-l\Prime)\lambda}\der\lambda,
			\end{aligned} \label{eq:C_inf}
		\end{equation}
		where we have used~\eqref{eq:wiener}.
			
		\subsection{Spectral ML detector}
		
			Knowing the statistical properties of the received signal allows us to represent the \acrshort{ml} metrics from~\eqref{eq:ml} in the spectral domain asymptotically.
			The following formulation reaches the \acrfull{cl} domain by first obtaining the \acrfull{kl} expansion of the involved signals and then applying a unitary transformation onto it~\cite{VilaInsa2024}.
			The data-dependent and logarithmic terms from~\eqref{eq:ml_terms} will be studied separately.
			While being structurally similar to their original expressions, the spectral counterparts of~\eqref{eq:ml_terms} will feature the interplay between signal properties much more visibly and allude to the formulation of a low-complexity receiver.
		
			\subsubsection{Data-dependent term}
			
				The received signal $\{\mathsf{y}_{i}(n)\}_{n\in\integs}$ is a second-order random process and, as such, admits an asymptotic spectral representation (as $\nr\to\infty$) known as \acrshort{kl} expansion~\cite[Sec.~9.1]{Schreier2010}.
				It is defined as
				\begin{equation}
					\mathsf{y}_{i}(n) \triangleq \int_{0}^{\mathrlap{1}}\phi_i(n,\lambda) 		\der\mathring{\mathsf{y}}_{i}(\lambda), \label{eq:kl_exp}
				\end{equation}
				for $n\in\integs$ and $\lambda\in[0,1)$, where $\{\der\mathring{\mathsf{y}}_i(\lambda)\}$ are uncorrelated increments in the \acrshort{kl} spectral domain and $\{\phi_i(n,\lambda)\}$ is a set of eigenfunctions that satisfy the following \textit{orthonormality} and \textit{completeness} relations~\cite[Ch.~2]{Kennedy2013}:
				\begin{equation}
					\int_{0}^{\mathrlap{1}}\phi(n,\lambda)\phi^*(n\Prime,\lambda)\der\lambda=\delta_{n-n\Prime},\quad \sum_{n\in\integs}\phi(n,\lambda)\phi^*(n,\lambda\Prime)=\delta(\lambda-\lambda\Prime). \label{eq:ortho}
				\end{equation}
				Each class of random processes has a different eigenbasis, obtained by solving the following eigenequation:
				\begin{equation}
					\smashoperator{\sum_{l\in\integs}}\infcov{i}(l,k-l)\phi_i(l,\lambda)=\klcyc{i}(\lambda)\phi_i(k,\lambda),\label{eq:eigen}
				\end{equation}
				where $\klcyc{i}(\lambda)\der\lambda \triangleq\expec[\lvert\der\mathring{\mathsf{y}}_i(\lambda)\rvert^2]$ is the \acrshort{kl} spectrum of $\{\mathsf{y}_{i}(n)\}$.
				Since we are dealing with a \acrshort{cs} process of period $K$, its \acrshort{kl} basis is
				\begin{equation}
					\phi_i^{(k)}(n,\sigma)\triangleq\phi_i\bigl(n,\sigma+\tfrac{k}{K}\bigr)= \smashoperator{\sum_{q=0}^{K-1}}[\mathbf{B}_i(\sigma)]_{q,k} \euler^{\imunit2\pi n(\sigma+\frac{q}{K})}\label{eq:cs_basis}
				\end{equation}
				for $\sigma\in[0,1/K)$ and $k\in\{0,\dots,K-1\}$, as established in~\cite{Riba2022,VilaInsa2024}.
				These are linear combinations of complex exponentials separated by multiples of the cyclic period.
				They generalize the \acrshort{kl} basis of \acrshort{wss} processes, which is known to be the set of Fourier eigenfunctions~\cite[Ch.~9]{Schreier2010}.
				The weights that dictate~\eqref{eq:cs_basis} are elements from $\mathbf{B}_i(\sigma)\in\complex^{K\times K}$, whose columns are the eigenvectors of the \textit{\acrfull{csm}} of the received signal, constructed as
				\begin{equation}
					[\cycmat{i}(\sigma)]_{r,c} \triangleq\cyc{i}^{(\frac{r-c}{K})}\bigl(\sigma+\tfrac{r}{K}\bigr), \label{eq:csm_def}
				\end{equation}
				where
				\begin{equation}
					\cyc{i}^{(\frac{k}{K})}(\lambda) \triangleq \tfrac{1}{K}\smashoperator{\sum_{n=0\vphantom{\integs}}^{K-1}}\smashoperator[r]{\sum_{m\in\integs}}\infcov{i}(n,m)\euler^{-\imunit2\pi(n\frac{k}{K}+m\lambda)} \label{eq:csf}
				\end{equation}
				is the \textit{cyclic spectrum} of $\{\mathsf{y}_{i}(n)\}$~\cite[Ch.~10]{Schreier2010}.
				Its eigendecomposition is then
				\begin{equation}
					\cycmat{i}(\sigma) \triangleq \mathbf{B}_i(\sigma)\klcycmat{i}(\sigma)\herm{\mathbf{B}}_i(\sigma) \in \complex^{K\times K}, \label{eq:csm_eig}
				\end{equation}
				from which we also obtain the \acrshort{kl} spectrum contained in $\klcycmat{i}(\sigma)\in\reals^{K\times K}$, a diagonal matrix constructed from the eigenvalues of $\cycmat{i}(\sigma)$:
				\begin{equation}
					[\klcycmat{i}(\sigma)]_{k,k} \equiv \klcyc{i}\bigl(\sigma+\tfrac{k}{K}\bigr) \triangleq \klcyc{i}^{(k)}(\sigma).
				\end{equation}
				
				Recall the data-dependent quadratic term of the \acrshort{ml} detector~\eqref{eq:ml_terms}:
				\begin{equation}
					\mathrm{Q}_j(\{\mathsf{y}_{i}(n)\}) = \tfrac{1}{K\nr} \smashoperator{\sum_{r,c=-\frac{K\nr}{2}}^{\frac{K\nr}{2}-1}} \mathsf{y}_{i}^*(r)[\covmat{j}^{-1}]_{r,c} \mathsf{y}_{i}(c). \label{eq:quadratic}
				\end{equation}
				The following result states its limit in the large array regime using the \acrshort{kl} spectral representation.
				\begin{proposition} \label{prop:quadratic}
					When $\nr\to\infty$, the quadratic term from~\eqref{eq:ml_terms} can be expressed in the \acrshort{kl} spectral domain as
					\begin{equation}
						\overset{\infty}{\mathrm{Q}}_j(\{\mathsf{y}_{i}(n)\}) \triangleq \lim_{\nr\to\infty} \mathrm{Q}_j(\{\mathsf{y}_{i}(n)\}) = \int_{0}^{\mathrlap{\frac{1}{K}}} \der\mathring{\boldsymbol{\mathsf{y}}} \herm{\vphantom{\boldsymbol{\mathsf{y}}}}_i(\sigma)\herm{\mathbf{B}}_i(\sigma)\cycmat{j}^{-1}(\sigma)\mathbf{B}_i(\sigma) \der\mathring{\boldsymbol{\mathsf{y}}}_i(\sigma), \label{eq:quadratic_kl}
					\end{equation}
					where $\der\mathring{\boldsymbol{\mathsf{y}}}_i(\sigma)\in\complex^{K}$ is a vector constructed by stacking $K$ samples of $\{\der\mathring{\mathsf{y}}_i(\sigma)\}$ spaced $\frac{1}{K}$, \ie
					\begin{equation}
						\der\mathring{\boldsymbol{\mathsf{y}}}_i(\sigma) \triangleq \trans{\bigl[\der\mathring{\mathsf{y}}_i(\sigma),\der\mathring{\mathsf{y}}_i\bigl(\sigma+\tfrac{1}{K}\bigr),\dots,\der\mathring{\mathsf{y}}_i\bigl(\sigma+\tfrac{K-1}{K}\bigr)\bigr]}.
					\end{equation}
				\end{proposition}
				\begin{proof}
					See \autoref{app:proof_quadratic}.
				\end{proof}
				
				Notice the structural resemblances between the original quadratic form~\eqref{eq:ml_terms} and the one in the \acrshort{kl} spectral domain~\eqref{eq:quadratic_kl}.
				We can now use results from~\cite{VilaInsa2024} to obtain it as a more compact expression.
				The \acrshort{cl} representation of $\{\mathsf{y}_{i}(n)\}$ is defined as
				\begin{equation}
					\der\breve{\mathsf{y}}_{i}(\lambda) \triangleq \lim_{\nr\to\infty} \tfrac{1}{K\nr} \smashoperator{\sum_{n=-\frac{K\nr}{2}}^{\frac{K\nr}{2}-1}}\mathsf{y}_{i}(n)\euler^{-\imunit2\pi n\lambda},
				\end{equation}
				which we can stack in a $K$-dimensional vector $[\der\breve{\boldsymbol{\mathsf{y}}}_i(\sigma)]_r\triangleq\der\breve{\mathsf{y}}_i(\sigma+\frac{r}{K})$.
				By~\cite[Corollary~1]{VilaInsa2024}, we know that $\der\breve{\boldsymbol{\mathsf{y}}}_i(\sigma) = \mathbf{B}_i(\sigma) \der\mathring{\boldsymbol{\mathsf{y}}}_i(\sigma)$.
				Therefore,~\eqref{eq:quadratic_kl} reduces to 
				\begin{equation}
					\overset{\infty}{\mathrm{Q}}_j(\{\mathsf{y}_{i}(n)\}) = \int_{0}^{\mathrlap{\frac{1}{K}}} \der\breve{\boldsymbol{\mathsf{y}}}\herm{\vphantom{\boldsymbol{\mathsf{y}}}}_i(\sigma) \cycmat{j}^{-1}(\sigma) \der\breve{\boldsymbol{\mathsf{y}}}_i(\sigma), \label{eq:Q_cl}
				\end{equation}
				which is the integral of a quadratic form in the \acrshort{cl} (Fourier) spectral domain.
				
				It is clear that, in order to analyze~\eqref{eq:Q_cl}, we need to obtain $\{\cycmat{j}(\sigma)\}$, which is stated next.
				\begin{proposition} \label{prop:csm}
					The \acrshort{csm} of $\{\mathsf{y}_j(n)\}$ from~\eqref{eq:y_i} under \acrshort{sswmod} is
					\begin{equation}
						\cycmat{j}(\sigma)= \herm{\mathbf{F}_K} \herm{\mathbf{\Theta}}(\sigma) \mathbf{X}_{j}\Ut\cycmat{\bsf{H}}(K\sigma)\herm{\Ut}\herm{\mathbf{X}_{j}}  \mathbf{\Theta}(\sigma) \mathbf{F}_K + \power{\bsf{Z}} \id{K}, \label{eq:csm}
					\end{equation}
					where $\mathbf{\Theta}(\sigma) \triangleq \Diag(1,\euler^{\imunit2\pi\sigma},\dots,\euler^{\imunit2\pi\sigma(K-1)})$.
				\end{proposition}
				\begin{proof}
					See \autoref{app:csm}.
				\end{proof}
				Let $\mathbf{G}_j(\sigma)\in\complex^{K\times K}$ be the signal contribution in~\eqref{eq:csm}, \ie $\mathbf{G}_j(\sigma)\triangleq\cycmat{j}(\sigma)-\power{\bsf{Z}}\id{K}$.
				Notice that its rank is $\leq\nt$.
				A similar idea was encountered in~\cite{Riba2010}, where it was shown that linearly modulated \acrfull{pam} signals exhibit a rank-1 structure in the \acrshort{cl} domain, leaving margin for noise power estimation in the null space in the case of oversampling.
				Hence, a clear duality manifests between the \textit{time} and \textit{combined time-space} domains.
				For instance, in a \acrshort{simo} scheme (see~\cite{VilaInsa2024b}) $\rank(\mathbf{G}_j(\sigma))=1$, implying that $\{\mathsf{y}_i(n)\}$ in~\eqref{eq:vec} would adopt a structure in combined time-space equivalent to a \acrshort{pam} in the time domain.
				For a general \acrshort{mimo} case with $\nt>1$, the signal received with the $\nr$ antennas is being treated herein as a linear combination of \acrshort{pam} signals, and this is what leads to an upper bound of $\nt$ for the rank of $\mathbf{G}_j(\sigma)$.
				
			\subsubsection{Logarithmic term}
				
				We have successfully expressed the data-dependent part of the \acrshort{ml} detector in the \acrshort{cl} spectral domain.
				To do so with the remaining logarithmic term, we will use some results on the asymptotic log-determinant of $K$-Toeplitz matrices~\cite{Widom1974}.
				It is known that
				\begin{equation}
					\overset{\infty}{\ell}_{j} \triangleq \lim_{\nr\to\infty} \ell_{j} = \lim_{\nr\to\infty} \tfrac{1}{K\nr}\ln\lvert\covmat{j}\rvert = \tfrac{1}{K}\int_{0}^{\mathrlap{1}} \ln\lvert\dot{\mathbf{C}}_{j}(\lambda)\rvert\der\lambda, \label{eq:asym_log_det}
				\end{equation}
				where
				\begin{equation}
					\dot{\mathbf{C}}_{j}(\lambda) \triangleq \smashoperator{\sum_{n\in\integs}} \covmat{j}(n)\euler^{-\imunit2\pi n\lambda}
				\end{equation}
				is the discrete-time Fourier transform of the $K\times K$ blocks defined in~\eqref{eq:blocks}.
				For the purposes of this work, it will be referred to as \textit{\acrfull{rsm}}~\cite[Sec.~1.4.1.3]{Hlawatsch2011}.
				We next state how it can be obtained from a unitary transformation of the \acrshort{csm} of $\{\mathsf{y}_{j}(n)\}$.
				\begin{proposition} \label{prop:log}
					The \acrshort{rsm} and \acrshort{csm} of a \acrshort{cs} process $\{\mathsf{y}_{j}(n)\}$ of period $K$ are related by
					\begin{equation}
						\cycmat{j}(\lambda) = \herm{\mathbf{F}}_K\herm{\mathbf{\Theta}}(\lambda)\dot{\mathbf{C}}_{j}(K\lambda)\mathbf{\Theta}(\lambda)\mathbf{F}_K.\label{eq:thm1}
					\end{equation}
				\end{proposition}
				\begin{proof}
					By the definition of the \acrshort{csm}~\eqref{eq:csm_def}, we know that
					\begin{equation}
						[\cycmat{j}(\lambda)]_{r,c} = \frac{1}{K} \sum_{n=0}^{K-1}\sum_{l\in\integs} \expec[\mathsf{y}_{j}(l)\mathsf{y}_{j}^*(n)] \euler^{\imunit2\pi(n(\lambda+\frac{c}{K})-l(\lambda+\frac{r}{K}))}.
					\end{equation}
					With some algebraic manipulations, we can express it in terms of the \acrshort{rsm}:
					\begin{align}
						[\cycmat{j}(\lambda)]_{r,c} &= \frac{1}{K} \smashoperator[r]{\sum_{n,p=0}^{K-1}} \euler^{\imunit2\pi(n(\lambda+\frac{c}{K})-p(\lambda+\frac{r}{K}))} \smashoperator{\sum_{l\in\integs}} \expec[\mathsf{y}_{j}(p+lK)\mathsf{y}_{j}^*(n)] \euler^{-\imunit2\pi lK(\lambda+\frac{r}{K})} \nonumber\\
						&= \frac{1}{K} \sum_{n,p=0}^{K-1} \euler^{-\imunit2\pi p(\lambda+\frac{r}{K})} [\dot{\mathbf{C}}_{j}(K\lambda)]_{p,n} \euler^{\imunit2\pi n(\lambda+\frac{c}{K})}.
					\end{align}
					We have used the fact that, by definition, the \acrshort{rsm} is periodic in $\lambda$ with period 1.
				\end{proof}
				
				Applying \autoref{prop:log} onto~\eqref{eq:asym_log_det}, we can represent the asymptotic logarithmic term as a function of the \acrshort{csm}:
				\begin{equation}
					\begin{aligned}
						\overset{\infty}{\ell}_{j} &= \tfrac{1}{K}\int_{0}^{\mathrlap{1}} \ln\bigl\lvert \mathbf{\Theta}\bigl(\tfrac{\lambda}{K}\bigr) \mathbf{F}_K \cycmat{j}\bigl(\tfrac{\lambda}{K}\bigr) \herm{\mathbf{F}}_K \herm{\mathbf{\Theta}}\big(\tfrac{\lambda}{K}\bigr) \bigr\rvert\der\lambda \\
						& = \tfrac{1}{K}\int_{0}^{\mathrlap{1}} \ln\bigl\lvert \cycmat{j}\bigl(\tfrac{\lambda}{K}\bigr)  \bigr\rvert\der\lambda = \int_{0}^{\mathrlap{\frac{1}{K}}} \ln\lvert \cycmat{j}(\sigma)  \rvert\der\sigma.
					\end{aligned}
				\end{equation}
				
			\subsubsection{ML metrics in the spectral domain}\label{sssec:ml_metrics}
				
				Using the results derived in the previous sections, we now display the total form of the asymptotic \acrshort{ml} metrics in the spectral domain:
				\begin{equation}
					\overset{\infty}{\mathcal{L}}_{j}(\{\mathsf{y}_{i}(n)\}) \triangleq \lim_{\nr\to\infty} \mathcal{L}_{j}(\vecrand{y}\vert\mathbf{X}_{i}) = \int_{0}^{\mathrlap{\frac{1}{K}}} 	\der\breve{\boldsymbol{\mathsf{y}}}\herm{\vphantom{\boldsymbol{\mathsf{y}}}}_i(\sigma) \cycmat{j}^{-1}(\sigma) \der\breve{\boldsymbol{\mathsf{y}}}_i(\sigma) + \ln\lvert\cycmat{j}(\sigma)\rvert\der\sigma. \label{eq:full_ml_spectral}
				\end{equation}
				Notice the structural similarities between this expression and the one from~\eqref{eq:ml}.
				
				In detection theory, the \acrfull{kld} between two hypotheses $\mathbf{X}_{i}$ and $\mathbf{X}_{j}$ plays a fundamental role in assessing their discrimination.
				It dictates the fastest error probability decay rate that can be achieved by any statistical test based on observations of $\{[\vecrand{y}\vert\mathbf{X}_i]_n\}$ as $\nr\to\infty$~\cite[Sec.~10.4]{Levy2008}.
				To obtain it, we first define the normalized \acrfull{llr} between the distributions of $\bsf{Y}\vert\mathbf{X}_{i}$ and $\bsf{Y}\vert\mathbf{X}_{j}$ as
				\begin{equation}
					\mathrm{L}_{i,j}^{(K\nr)}(\vecdet{y})\triangleq\tfrac{1}{K\nr}\ln\frac{\mathrm{f}_{\bsf{Y}\vert\mathbf{X}_i}(\vecdet{y})}{\mathrm{f}_{\bsf{Y}\vert\mathbf{X}_j}(\vecdet{y})}. \label{eq:llr}
				\end{equation}
				Then, the normalized \acrshort{kld} is
				\begin{equation}
					\mathcal{D}_{\mathrm{KL}}(i\Vert j) \triangleq \expec\bigl[ \mathrm{L}_{i,j}^{(K\nr)}(\vecrand{y}\vert\mathbf{X}_{i}) \bigr] = \expec\bigl[\mathcal{L}_{j}(\vecrand{y}\vert\mathbf{X}_{i}) - \mathcal{L}_{i}(\vecrand{y}\vert\mathbf{X}_{i})\bigr], \label{eq:kld}
				\end{equation}
				\ie the expectation of their \acrshort{llr} when $\mathbf{X}_i$ has been transmitted.
				It yields the value towards which $\mathrm{L}_{i,j}^{(K\nr)}(\vecrand{y}\vert\mathbf{X}_{i})$ converges asymptotically (in the mean-square sense), so larger values of $\mathcal{D}_{\mathrm{KL}}(i\Vert j)$ are usually linked to improved separation of the two hypotheses.
				
				Having derived~\eqref{eq:full_ml_spectral}, we can state the normalized \acrshort{kld} asymptotically in the \acrshort{cl} spectral domain:
				\begin{equation}
					\begin{aligned}
						\overset{\infty}{\mathcal{D}}_{\mathrm{KL}}(i\Vert j) &\triangleq \lim_{\nr\to\infty} \mathcal{D}_{\mathrm{KL}}(i\Vert j) = \expec\Bigl[\overset{\infty}{\mathrm{Q}}_{j}(\{\mathsf{y}_{i}(n)\})-\overset{\infty}{\mathrm{Q}}_{i}(\{\mathsf{y}_{i}(n)\})\Bigr] + \overset{\infty}{\ell}_{j} - \overset{\infty}{\ell}_{i} \\
						&= \int_{0}^{\mathrlap{\frac{1}{K}}} 	\expec\bigl[\der\breve{\boldsymbol{\mathsf{y}}}\herm{\vphantom{\boldsymbol{\mathsf{y}}}}_i(\sigma) \bigl( \cycmat{j}^{-1}(\sigma) - \cycmat{i}^{-1}(\sigma)\bigl) \der\breve{\boldsymbol{\mathsf{y}}}_i(\sigma)\bigr] + \ln\lvert\cycmat{j}(\sigma)\rvert-\ln\lvert\cycmat{i}(\sigma)\rvert\der\sigma.
					\end{aligned}
				\end{equation}
				The first term is obtained by using the cyclic property and linearity of the trace, as well as some identities from~\cite{VilaInsa2024}:
				\begin{equation}
					\overset{\infty}{\mathcal{D}}_{\mathrm{KL}}(i\Vert j) = \int_{0}^{\mathrlap{\frac{1}{K}}} \der\sigma \biggl(\trace\bigl[ \cycmat{j}^{-1}(\sigma)\cycmat{i}(\sigma) - \id{K}\bigr] - \ln\frac{\lvert\cycmat{i}(\sigma)\rvert}{\lvert\cycmat{j}(\sigma)\rvert}\biggr). \label{eq:kld_spectral}
				\end{equation}
				As expected, this expression is reminiscent of the \acrshort{kld} between two zero-mean multivariate normal distributions, as noted by the structure of its integrand~\cite[Sec.~3.2]{Levy2008}.
				Representing $\{\cycmat{i}(\sigma),\cycmat{j}(\sigma)\}$ in terms of their eigendecompositions (through the use of \autoref{prop:csm}), will prove useful in \autoref{sec:structure} to examine the fundamental aspects of the codeword detection problem.
			
		\subsection{Low-complexity implementation of ML detection} \label{ssec:low_complex}
		
			The main advantage of using the spectral representation of the \acrshort{ml} detector derived previously is that it can be approximated efficiently for finite $\nr$.
			Throughout this section, we assume the receiver is aware of $\{\covmat{j}\}$ for all  $\mathbf{X}_{j}\in\mathcal{X}$.
			The acquisition of this second-order statistical \acrshort{csi} is a topic of active research and outside the scope of this work~\cite{Yang2023}.
			Nevertheless, it being available at the \acrshort{bs} is a reasonable premise due to its slow-varying nature\footnote{This long-term \acrshort{csi} includes those parameters related to the channel geometry.}~\cite[Sec.~4.5]{Heath2018} and the fact it can be estimated in the presence of data~\cite{Kim2008}.
			
			\subsubsection{Received data} \label{sssec:rx_data}
			
				To implement the \acrshort{ml} detector in the spectral domain, the receiver must first obtain the \acrshort{cl} expansion of the received signal, $\{\der\breve{\boldsymbol{\mathsf{y}}}_i(\sigma)\}$.
				It can be approximated by
				\begin{equation}
					\bigl[\der\widehat{\boldsymbol{\mathsf{y}}}_i\bigl(\tfrac{l}{K\nr}\bigr)\bigr]_k \triangleq \tfrac{1}{K\nr}\fft_{K\nr}[\mathsf{y}_i(n)](l+k\nr), \label{eq:fft_cl}
				\end{equation}
				where $\fft_{K\nr}[\placeholder]$ is the \acrfull{fft} of $K\nr$ points.
				Therefore, the set of $\{\der\widehat{\boldsymbol{\mathsf{y}}}_i(\frac{l}{K\nr})\}_{l\in\{0,\dots,\nr-1\}}$ can be computed with $O(K\nr\ln(K\nr))$ operations.
				Then, the quadratic form in~\eqref{eq:full_ml_spectral} is approximated as
				\begin{equation}
					\widehat{\mathrm{Q}}_j(\{\mathsf{y}_{i}(n)\}) \triangleq \smashoperator{\sum_{l=0}^{\nr-1}} \der\widehat{\boldsymbol{\mathsf{y}}}\herm{\vphantom{\boldsymbol{\mathsf{y}}}}_i\bigl(\tfrac{l}{K\nr}\bigr) \widehat{\mathbf{S}}_{j}^{-1}\bigl(\tfrac{l}{K\nr}\bigr) \der\widehat{\boldsymbol{\mathsf{y}}}_i\bigl(\tfrac{l}{K\nr}\bigr) \label{eq:fft_Q}
				\end{equation}
				with $O(\nr K^2)$ operations, where $\widehat{\mathbf{S}}_{j}\bigl(\frac{l}{K\nr}\bigr)$ is an estimation of $\cycmat{j}\bigl(\frac{l}{K\nr}\bigr)$.
			
			\subsubsection{Statistics} \label{sssec:statistics}
			
				Beside processing the received data with~\eqref{eq:fft_cl} and~\eqref{eq:fft_Q}, approximating the \acrshort{ml} metrics in spectral form requires the computation of various statistical structures from $\{\covmat{j}\}$.
				While direct implementation of the \acrshort{ml} detector~\eqref{eq:ml} demands $\{\covmat{j}^{-1}\}$, with a cost of $O(K^3\nr^3)$ operations for each $\mathbf{X}_{j}\in\mathcal{X}$, the spectral \acrshort{ml} approximation uses inverse \acrshort{csm}s.
				
				The set of $\{\widehat{\mathbf{S}}_{j}^{-1}(\frac{l}{K\nr})\}$ can be obtained as follows.
				Recall the \acrshort{rsm} is computed as
				\begin{equation}
					\dot{\mathbf{C}}_{j}(\lambda) = \smashoperator{\sum_{n=0}^{\infty}} \covmat{j}(n) \euler^{-\imunit2\pi n\lambda} + \smashoperator{\sum_{n\Prime=-\infty}^{-1}} \covmat{j}(n\Prime) \euler^{-\imunit2\pi n\Prime\lambda}.
				\end{equation}
				With some simple algebraic manipulations and considering~\eqref{eq:herm}, it may be obtained as
				\begin{equation}
					\dot{\mathbf{C}}_{j}(\lambda) = \smashoperator{\sum_{n=0}^{\infty}} \covmat{j}(n) \euler^{-\imunit2\pi n\lambda} + \herm{\biggl(\smashoperator[r]{\sum_{n\Prime=0}^{\infty}} \covmat{j}(n\Prime) \euler^{-\imunit2\pi n\Prime\lambda}\biggr)} - \covmat{j}(0).
				\end{equation}
				Therefore, it is straightforward to see it can be approximated as
				\begin{equation}
					\widehat{\dot{\mathbf{C}}}_{j}\bigl(\tfrac{l}{\nr}\bigr) \triangleq \fft_{\nr}[\covmat{j}(n)](l) + \fft_{\nr}\herm{[\covmat{j}(n)]}(l) - \covmat{j}(0), \label{eq:rsm_impl}
				\end{equation}
				with which the \acrshort{csm} is estimated as
				\begin{equation}
					\widehat{\mathbf{S}}_j\bigl(\tfrac{l}{K\nr}\bigr) \triangleq \herm{\mathbf{F}}_K\herm{\mathbf{\Theta}}\bigl(\tfrac{l}{K\nr}\bigr)\widehat{\dot{\mathbf{C}}}_{j}\bigl(\tfrac{l}{\nr}\bigr)\mathbf{\Theta}\bigl(\tfrac{l}{K\nr}\bigr)\mathbf{F}_K,\label{eq:fft_S}
				\end{equation}
				based on \autoref{prop:log}.
				Notice that this expression yields a hermitian matrix but does not guarantee it to be positive semi-definite\footnote{We know from~\cite{VilaInsa2024} that $\cycmat{j}(\sigma)$ is a covariance matrix and thus, positive semi-definite.} in general for finite $\nr$.
				Therefore, following the same idea behind the \textit{Blackman--Tukey spectral estimator}~\cite[Sec.~4.5]{Kay1988}, we weight the correlation blocks in~\eqref{eq:rsm_impl} with a triangular window:
				\begin{equation}
					\widehat{\dot{\mathbf{C}}}\vphantom{\mathbf{C}}_{j}^{(\mathrm{BT})}\bigl(\tfrac{l}{\nr}\bigr) = \fft_{\nr}\bigl[\bigl(1-\tfrac{n}{\nr}\bigr)\covmat{j}(n)\bigr](l) + \fft_{\nr}\herm{\bigl[\bigl(1-\tfrac{n}{\nr}\bigr)\covmat{j}(n)\bigr]}(l) - \covmat{j}(0). \label{eq:rsm_bt}
				\end{equation}
				
				In summary, computing $\{\widehat{\mathbf{S}}_{j}^{-1}(\frac{l}{K\nr})\}_{l\in\{0,\dots,\nr-1\}}$ involves $O(K^2\nr\ln\nr)$ operations for the \acrshort{rsm}, plus $O(K^3\nr)$ additional for~\eqref{eq:fft_S} and the matrix inversions.
				Fortunately, these statistics can be pre-calculated at the receiver before data transmission and are valid for the detection of a certain number $T$ of consecutive codewords, in which the channel statistics remain stable~\cite[Sec.~4.5]{Heath2018}.
				The same can be said for the logarithmic terms, whose direct computation involves $O(K^3\nr^3)$ operations each, whereas their spectral approximations
				\begin{equation}
					\widehat{\ell}_{j} = \smashoperator{\sum_{l=0}^{\nr-1}}\tfrac{1}{K\nr} \ln\bigl\lvert\widehat{\mathbf{S}}_{j}\bigl(\tfrac{l}{K\nr}\bigr)\bigr\rvert
				\end{equation}
				are computed with just $O(K^3\nr)$.
				
			\subsubsection{Approximation validity}\label{sssec:approx}
			
				\begin{table}[ht]
					\centering
					\normalsize
					\caption{Computational complexity of \acrshort{ml} detection.}\label{tab:complex}
					\begin{tabular}{| c | c | c |}
						\toprule
						\multicolumn{3}{|c|}{\textbf{Once every $T$ codewords}} \\
						\midrule
						& \textbf{Direct \acrshort{ml}} & \textbf{Spectral \acrshort{ml} approx.} \\
						\midrule
						$\{\covmat{j}^{-1}\}\,\vert\,\{\cycmat{j}^{-1}(\sigma)\}$ & $O(MK^3\nr^3)$ & $O(MK^2\nr(\ln\nr+K))$ \\
						\midrule
						$\{\ell_{j}\}$ & $O(MK^3\nr^3)$ & $O(MK^3\nr)$ \\
						\midrule
						\midrule
						\multicolumn{3}{|c|}{\textbf{Once for each codeword detection}} \\
						\midrule
						& \textbf{Direct \acrshort{ml}} & \textbf{Spectral \acrshort{ml} approx.} \\
						\midrule
						$\{\mathrm{Q}_j\}$ & $O(MK^2\nr^2)$ & $O(K\nr(\ln(K\nr) + MK))$ \\
						\midrule
						\acrshort{ml} detection & $O(M)$ & $O(M)$ \\
						\botrule
					\end{tabular}
				\end{table}
				\autoref{tab:complex} contains a summary of the computational complexity associated to each implementation.
				The advantages of using the spectral approximation are clear when employing large arrays at the \acrshort{bs}, since the number of operations to compute both data statistics and \acrshort{ml} metrics scale slower with $\nr$.
				Nevertheless, this approach is based on an asymptotic representation of the \acrshort{ml} detector as $\nr\to\infty$, and its validity for finite $\nr$ remains to be seen.
				
				The performance of the spectral \acrshort{ml} approximation described in the previous section will be assessed next against the regular \acrshort{ml} detector from \autoref{ssec:uml}.
				We consider a channel under \acrshort{sswmod}, whose \acrshort{acf} is constructed from~\eqref{eq:channel_acf}, by setting
				\begin{equation}
					[\covmat{\bsf{H}}(m)]_{k,k} \triangleq \tfrac{1}{k+2} \sinc^2\bigl(\tfrac{m}{k+2}\bigr), \quad k=0,\dots,\nt-1.
				\end{equation}
				This yields a set of spectra $\{\cyc{\bsf{H}}^{(k)}(\lambda)\}$ with triangular shape of widths $\{\frac{1}{k+2}\}$.
				This fading profile has been adopted motivated by two main reasons.
				On the one hand, it represents a simple model constructed from piece-wise linear functions of compact support.
				On the other hand, it characterizes a sufficiently well-behaved, worst-case scenario.
				Since triangular functions belong to differentiability class $C^0$, the correlation drop-off between adjacent antennas is assumed to decay quadratically, much slower than one corresponding to smooth spectra~\cite[Sec.~4.4.5]{Vetterli2014}.
				Note, however, that $\{\cyc{\bsf{H}}^{(k)}(\lambda)\in C^0\}$ avoids the occurrence of disturbing artifacts such as Gibbs phenomenon.
				
				The channel correlation matrix is then
				\begin{equation}
					\covmat{\vecrand{h}} = \eta\cdot(\id{\nr} \otimes \Ut)\begin{bmatrix}
						\covmat{\bsf{H}}(0) & \covmat{\bsf{H}}(1) & \cdots & \covmat{\bsf{H}}(\nr-1) \\
						\covmat{\bsf{H}}(1) & \covmat{\bsf{H}}(0) & \cdots & \covmat{\bsf{H}}(\nr-2) \\
						\vdots & \vdots & \ddots & \vdots \\
						\covmat{\bsf{H}}(\nr-1) & \covmat{\bsf{H}}(\nr-2) & \cdots & \covmat{\bsf{H}}(0)
					\end{bmatrix} (\id{\nr} \otimes \herm{\Ut}),
				\end{equation}
				where $\eta$ is a power normalization parameter that ensures that $\trace[\covmat{\vecrand{h}}]\triangleq\nt\nr$, as stated in the signal model from \autoref{ssec:model}.
				The transmitter basis $\Ut$ is selected randomly from the set of unitary $\nt\times\nt$ matrices~\cite{Mezzadri2007}.
				Noise power is $\frac{1}{K}\expec[\lVert\bsf{Z}\rVert_{\frob}^2]=\power{\bsf{Z}}$\nr.
				\acrshort{snr} at the receiver is defined as
				\begin{equation}\label{eq:snr}
					\snr \triangleq \frac{\expec[\lVert\bsf{XH}\rVert_{\frob}^2]}{\expec[\lVert\bsf{Z}\rVert_{\frob}^2]} = \frac{\trace\Bigl[\covmat{\vecrand{h}}\expec\bigl[\herm{\vecrand{X}\vphantom{\bsf{X}}}\vecrand{X}\bigr]\Bigr]}{\power{\bsf{Z}}K\nr}.
				\end{equation}
				
				Setting $K\geq2\nt$, the pair of tested codewords is constructed as follows.
				An element $\mathbf{U}\in\complex^{K\times K}$ is selected randomly from the set of unitary $K\times K$ matrices.
				Its first $\nt$ columns are named $\mathbf{U}_i\in\complex^{K\times\nt}$, and the following $\nt$ ones are named $\mathbf{U}_j\in\complex^{K\times\nt}$.
				Then, codewords are $\mathbf{X}_i \triangleq \sqrt{\frac{K}{\nt}}\mathbf{U}_i$ and $\mathbf{X}_j \triangleq \sqrt{\frac{K}{2\nt}} \mathbf{U}_j$, thus differing both in power ($1$ and $1/2$, respectively) and spanned subspace (they are orthogonal), as to illustrate the effects caused by both discrepancies.
				Under this transmission scheme,~\eqref{eq:snr} simplifies to $\snr = \frac{3/4}{\power{\bsf{Z}}}$.
				
				\begin{figure}[t]
					\centering
					\resizebox{\textwidth}{!}{
						\begin{tikzpicture}
							\pgfplotsset{set layers}
							\pgfplotstableread[col sep=comma, row sep=newline]{Nr-10.csv}{\tauladata}
							\pgfplotstabletranspose[colnames from=Setting]\taula{\tauladata}
							\pgfplotstableread[col sep=comma, row sep=newline]{Nr-5.csv}{\tauladata}
							\pgfplotstabletranspose[colnames from=Setting]\taulaa{\tauladata}
							\pgfplotsset{mlline/.style={thick, mark=x, dotted, mark options={solid}, mark size=4pt, forget plot}}
							\pgfplotsset{spectralline/.style={thick, solid, mark options={solid}, mark size=4pt}}
							\begin{loglogaxis}[
								width=\linewidth,	height=.5\linewidth,
								xmin=2,			ymin=5e-8,
								xmax=256,		ymax=1,
								xlabel={$\nr$},		ylabel=$\text{Pr}(\text{error})$,
								log basis x=2,		log basis y=10,
								axis line style={thick},
								legend pos = south west,
								legend style={on layer = axis foreground},
								clip marker paths=true,
								legend cell align=left,
								grid=major,
								ytick distance=10
								]
								\addplot [spectralline, RoyalBlue1, mark=+] table [y=K4spec] {\taula};
								\addlegendentry{$K=4$, $\nt=1$};
								\addplot [mlline, RoyalBlue1] table [y=K4] {\taula};
								\addplot [spectralline, SpringGreen3, mark=oplus] table [y=K6spec] {\taula};
								\addlegendentry{$K=6$, $\nt=2$};
								\addplot [mlline, SpringGreen3] table [y=K6] {\taula};
								\addplot [spectralline, Firebrick1, mark=-] table [y=K8spec] {\taula};
								\addlegendentry{$K=8$, $\nt=4$};
								\addplot [mlline, Firebrick1] table [y=K8] {\taula};
								\addplot [thick, mark=x, dotted, mark options={solid}, mark size=4pt] coordinates {(1,1)};
								\addlegendentry{Regular \acrshort{ml}};
								\node[shape=ellipse, draw, minimum width=.5cm, minimum height=1.25cm, dashed] (ell1) at (100,0.001) {};
								\node[anchor=south west] at (ell1.north) {\qty{-10}{\dB}};
								\addplot [spectralline, RoyalBlue1, mark=+] table [y=K4spec] {\taulaa};
								\addplot [mlline, RoyalBlue1] table [y=K4] {\taulaa};
								\addplot [spectralline, SpringGreen3, mark=oplus] table [y=K6spec] {\taulaa};
								\addplot [mlline, SpringGreen3] table [y=K6] {\taulaa};
								\addplot [spectralline, Firebrick1, mark=-] table [y=K8spec] {\taulaa};
								\addplot [mlline, Firebrick1] table [y=K8] {\taulaa};
								\node[shape=ellipse, draw, minimum width=.25cm, minimum height=1cm, dashed] (ell2) at (12,0.007) {};
								\node[anchor=north east] at (ell2.south) {\qty{-5}{\dB}};
							\end{loglogaxis}
						\end{tikzpicture}
					}
					\caption{Error probability in terms of $\nr$, for various $K$ and $\nt$ at $\snr=$\qty{-10}{\dB} and $\snr=$\qty{-5}{\dB}.}
					\label{fig:setting_1}
				\end{figure}
				\begin{figure}[t]
					\centering
					\resizebox{\textwidth}{!}{
					\begin{tikzpicture}
						\pgfplotsset{set layers}
						\pgfplotstableread[col sep=comma, row sep=newline]{SNR32.csv}{\tauladata}
						\pgfplotstabletranspose[colnames from=Setting]\taula{\tauladata}
						\pgfplotstableread[col sep=comma, row sep=newline]{SNR128.csv}{\tauladata}
						\pgfplotstabletranspose[colnames from=Setting]\taulaa{\tauladata}
						\pgfplotsset{mlline/.style={thick, mark=x, dotted, mark options={solid}, mark size=4pt, forget plot}}
						\pgfplotsset{spectralline/.style={thick, solid, mark options={solid}, mark size=4pt}}
						\begin{semilogyaxis}[
							width=\linewidth,	height=.5\linewidth,
							xmin=-15,			ymin=5e-8,
							xmax=-2,			ymax=1,
							xlabel={\acrshort{snr} [\unit{\dB}]},		ylabel=$\text{Pr}(\text{error})$,
							log basis y=10,
							axis line style={thick},
							legend pos = south west,
							legend style={on layer = axis foreground},
							clip marker paths=true,
							legend cell align=left,
							grid=major,
							ytick distance=10
							]
							\addplot [spectralline, RoyalBlue1, mark=+] table [y=K4spec] {\taula};
							\addlegendentry{$K=4$, $\nt=1$};
							\addplot [mlline, RoyalBlue1] table [y=K4] {\taula};
							\addplot [spectralline, SpringGreen3, mark=oplus] table [y=K6spec] {\taula};
							\addlegendentry{$K=6$, $\nt=2$};
							\addplot [mlline, SpringGreen3] table [y=K6] {\taula};
							\addplot [spectralline, Firebrick1, mark=-] table [y=K8spec] {\taula};
							\addlegendentry{$K=8$, $\nt=4$};
							\addplot [mlline, Firebrick1] table [y=K8] {\taula};
							\addplot [thick, mark=x, dotted, mark options={solid}, mark size=4pt] coordinates {(-16,1)};
							\addlegendentry{Regular \acrshort{ml}};
							\node[shape=ellipse, draw, minimum width=.5cm, minimum height=1.75cm, dashed] (ell1) at (-5,0.00005) {};
							\node[anchor=south west] at (ell1.north) {32};
							\addplot [spectralline, RoyalBlue1, mark=+] table [y=K4spec] {\taulaa};
							\addplot [mlline, RoyalBlue1] table [y=K4] {\taulaa};
							\addplot [spectralline, SpringGreen3, mark=oplus] table [y=K6spec] {\taulaa};
							\addplot [mlline, SpringGreen3] table [y=K6] {\taulaa};
							\addplot [spectralline, Firebrick1, mark=-] table [y=K8spec] {\taulaa};
							\addplot [mlline, Firebrick1] table [y=K8] {\taulaa};
							\node[shape=ellipse, draw, minimum width=.5cm, minimum height=1.75cm, dashed] (ell2) at (-9,0.00002) {};
							\node[anchor=north east] at (ell2.south) {128};
						\end{semilogyaxis}
					\end{tikzpicture}
					}
					\caption{Error probability in terms of \acrshort{snr}, for various $K$ and $\nt$, with $\nr=32$ and $\nr=128$.}
					\label{fig:setting_2}
				\end{figure}
				
				\autoref{fig:setting_1} and~\autoref{fig:setting_2} assess the performance of detectors from \autoref{ssec:uml} (Regular \acrshort{ml}) and \autoref{ssec:low_complex} (Spectral \acrshort{ml}) for different configurations of $K$ and $\nt$, at various regimes of $\nr$ and \acrshort{snr}.
				The metric of choice is the average error probability in the binary detection problem involving $\mathbf{X}_i$ and $\mathbf{X}_j$, computed as
				\begin{equation}
					\text{Pr}(\text{error}) \triangleq \frac{\text{Pr}\bigl(\text{detect }\mathbf{X}_j\vert\text{transmit }\mathbf{X}_i\bigr)+\text{Pr}\bigl(\text{detect }\mathbf{X}_i\vert\text{transmit }\mathbf{X}_j\bigr)}{2}.
				\end{equation}
				It has been averaged over \num{2e7} Monte Carlo simulations.
				
				On the one hand, \autoref{fig:setting_1} contains the detection error probability for increasing values of $\nr$, at $\snr=$\qty{-10}{\dB} and $\snr=$\qty{-5}{\dB}.
				All six curves display similar decreasing trends, with overall better results for higher $K$ and $\nt$.
				Remarkably, the difference between regular and spectral \acrshort{ml} detectors is diminutive, even for moderate receiver arrays.
				It is most noticeable for $\nr<2^5$ in the higher \acrshort{snr} regime.
				On the other hand, \autoref{fig:setting_2} shows the same metric in terms of \acrshort{snr}, for $\nr=32$ and $\nr=128$.
				The performance loss incurred from using the spectral approximation is only noticeable for the smaller array at lower error probability values (higher \acrshort{snr}).
				Based on these results, we can state the low-complexity detector developed through \autoref{ssec:low_complex} can be used instead of~\eqref{eq:ml} in a wide range of configurations (particularly with large arrays) with negligible performance degradation.
				
	\section{Codeword structure analysis} \label{sec:structure}
		
		In this section we will get some final insights on the problem considered in this work, as well as connect it to various ideas from the literature on detection theory.
		The figure of merit will be the asymptotic normalized \acrshort{kld}, whose relevance in evaluating hypothesis testing has been discussed in \autoref{sssec:ml_metrics}.
		More precisely, it will be linked to various results on singular (\ie error-free) detection: understanding if two codewords can be detected unmistakably allows to characterize the existence of error floors under various asymptotic regimes.
		This is of particular interest when assessing the gains achievable by pouring more resources into a system.
		
		We study the transmission of structured codewords within the presented spectral domain.
		Each transmitted codeword, taken from $M$-ary alphabet $\mathcal{X}$, can be constructed from its singular value decomposition as $\mathbf{X}_{i}\triangleq\mathbf{\Phi}_{i}\mathbf{W}_{i}^{\frac{1}{2}}\herm{\mathbf{\Psi}_{i}}$, where:
		\begin{itemize}
			\item $\mathbf{\Phi}_{i}\in\complex^{K\times K}$ is a unitary matrix.
			\item $\mathbf{W}_{i}^{\frac{1}{2}}\in\reals^{K\times\nt}$ is a power-loading diagonal matrix with nonnegative real entries $\{[\mathbf{W}_{i}^{\frac{1}{2}}]_{p,p}\triangleq (w_{i}^{(p)})^{\frac{1}{2}}\}_{p=0,\dots,\nt-1}$.
			The first $N_i\leq\nt$ ones are non-null (\textit{active}), while the rest are 0 (\textit{inactive}).
			\item $\mathbf{\Psi}_{i}\in\complex^{\nt\times\nt}$ is a unitary precoding matrix.
		\end{itemize}
		The rank of $\mathbf{X}_{i}$ is $N_i$, as stated by the definition of $\mathbf{W}_{i}^{\frac{1}{2}}$.
		Therefore, $\mathbf{\Phi}_{i}\mathbf{W}_{i}^{\frac{1}{2}}$ spans a $N_i$-dimensional subspace $\mathrm{V}_i$.
		This formulation allows us to straightforwardly get the \acrshort{csm} basis $\mathbf{B}_i(\sigma)$  in some relevant scenarios that will be explored next.
			
		\subsection{Statistical \texorpdfstring{\acrshort{csi}}{} at the transmitter} \label{ssec:csit}

			When the transmitter is aware of $\Ut$ in~\eqref{eq:weichselberger} (\ie the channel geometry), it can implement a precoding suitable to the fading correlation geometry: $\mathbf{\Psi}_{i} \triangleq \Ut$, which is to be understood as transmitting along the channel basis.
			This reduces the difficulty of eigendecomposing $\cycmat{i}(\sigma)$ immensely.
			We can clearly see from~\eqref{eq:csm} that its eigenbasis and \acrshort{kl} spectrum are
			\begin{equation}
				\mathbf{B}_i(\sigma) = \herm{\mathbf{F}}_{K} \herm{\mathbf{\Theta}}(\sigma) \mathbf{\Phi}_{i}, \quad \klcycmat{i}(\sigma) = \mathbf{W}_{i}^{\frac{1}{2}}\cycmat{\bsf{H}}(K\sigma)\herm{\bigl(\mathbf{W}_{i}^{\frac{1}{2}}\bigr)}+\power{\bsf{Z}}\id{K}.
			\end{equation}
			The normalized \acrshort{kld} resulting from using this configuration onto~\eqref{eq:kld_spectral} is
			\begin{equation}
				\overset{\infty}{\mathcal{D}}_{\mathrm{KL}}(i\Vert j) = \int_{0}^{\mathrlap{\frac{1}{K}}} \der\sigma \biggl(\trace\bigl[\mathbf{\Phi}_{j}\klcycmat{j}^{-1}(\sigma)\herm{\mathbf{\Phi}_{j}}\mathbf{\Phi}_{i}\klcycmat{i}(\sigma)\herm{\mathbf{\Phi}_{i}} - \id{K}\bigr] - \ln\frac{\lvert\klcycmat{i}(\sigma)\rvert}{\lvert\klcycmat{j}(\sigma)\rvert}\biggr).
			\end{equation}

			In the case in which $\mathbf{\Phi}_{i}=\mathbf{\Phi}_{j}$, \ie both codewords use the same signal basis, information is purely modulated on the transmitted energy through $\mathbf{W}_i$ and $\mathbf{W}_j$.
 			The asymptotic normalized \acrshort{kld} simplifies into
 			\begin{equation}
 				\overset{\infty}{\mathcal{D}}_{\mathrm{KL}}(i\Vert j) = \int_{0}^{\mathrlap{\frac{1}{K}}} \der\sigma \biggl(\trace\bigl[\klcycmat{j}^{-1}(\sigma)\klcycmat{i}(\sigma) - \id{K}\bigr] - \ln\frac{\lvert\klcycmat{i}(\sigma)\rvert}{\lvert\klcycmat{j}(\sigma)\rvert}\biggr) \triangleq \smashoperator{\sum_{k=0}^{\nt-1}} \mathcal{D}_{\mathrm{IS}}^{(k)}(\klcyc{i}\Vert\klcyc{j}), \label{eq:is_div}
 			\end{equation}
 			where
 			\begin{equation}
 				\mathcal{D}_{\mathrm{IS}}^{(k)}(\klcyc{i}\Vert\klcyc{j}) \triangleq \int_{0}^{\mathrlap{\frac{1}{K}}} \der\sigma \Biggl(\frac{w_{i}^{(k)}\cyc{\bsf{H}}^{(k)}(K\sigma)+\power{\bsf{Z}}}{w_{j}^{(k)}\cyc{\bsf{H}}^{(k)}(K\sigma)+\power{\bsf{Z}}} - 1 - \ln\frac{w_{i}^{(k)}\cyc{\bsf{H}}^{(k)}(K\sigma)+\power{\bsf{Z}}}{w_{j}^{(k)}\cyc{\bsf{H}}^{(k)}(K\sigma)+\power{\bsf{Z}}}\Biggr). \label{eq:is_div_k}
 			\end{equation}
			This is the \textit{\acrfull{isd}}~\cite[Sec.~10.4.2]{Levy2008} between $\klcyc{i}^{(k)}(\sigma)$ and $\klcyc{j}^{(k)}(\sigma)$ in the $k$th \acrshort{kl} sub-band.
			This measure of dissimilarity is nonnegative and only null when $\klcyc{i}^{(k)}(\sigma)=\klcyc{j}^{(k)}(\sigma)$ (except in a support of measure 0).
			Moreover,~\eqref{eq:is_div} is the \acrshort{isd} between the full \acrshort{kl} spectra corresponding to $\{\mathsf{y}_i(n)\}$ and $\{\mathsf{y}_j(n)\}$:
			\begin{equation}
				\mathcal{D}_{\mathrm{IS}}(\klcyc{i}\Vert\klcyc{j}) \triangleq \int_{0}^{\mathrlap{1}}\der\lambda \biggl(\frac{\klcyc{i}(\lambda)}{\klcyc{j}(\lambda)} - 1 - \ln\frac{\klcyc{i}(\lambda)}{\klcyc{j}(\lambda)}\biggr).
			\end{equation}
			
			It is of interest to study the tendency of~\eqref{eq:is_div_k} in the high \acrshort{snr} regime, \ie as $\power{\bsf{Z}}\to0$ (see~\eqref{eq:snr}), since it reveals the properties of $\mathbf{X}_i$ and $\mathbf{X}_j$ that are relevant for their detection asymptotically.
			Given $N_j>N_i$, the corresponding \acrshort{isd} in the $k$th \acrshort{kl} sub-band for $N_i\leq k<N_j$ is
			\begin{equation}\label{eq:high_snr_energy_1}
				\lim_{\power{\bsf{Z}}\to0} \mathcal{D}_{\mathrm{IS}}^{(k)}(\klcyc{i}\Vert\klcyc{j}) \geq -\tfrac{1}{K} - \int_{0}^{\mathrlap{\frac{1}{K}}} \der\sigma \biggl( \lim_{\power{\bsf{Z}}\to0}\ln\frac{\power{\bsf{Z}}}{w_{j}^{(k)}\cyc{\bsf{H}}^{(k)}(K\sigma)+\power{\bsf{Z}}}\biggr) = +\infty,
			\end{equation}
			where we have applied Fatou's Lemma~\cite[Lemma~B.2.4]{Choudary2014} in order to swap the order of limit and integration.
			The same divergence is observed when $N_i>N_j$ for $N_j\leq k<N_i$.
			On the contrary, when $N_i=N_j$,
			\begin{equation}\label{eq:high_snr_energy}
				\lim_{\power{\bsf{Z}}\to0} \overset{\infty}{\mathcal{D}}_{\mathrm{KL}}(i\Vert j) = \tfrac{1}{K} \sum_{k=0}^{N_i-1} \frac{w_{i}^{(k)}}{w_{j}^{(k)}} - 1 - \ln\frac{w_{i}^{(k)}}{w_{j}^{(k)}} 
			\end{equation}
			by \acrfull{ldct}~\cite[Th.~11.3.13]{Choudary2014}.
			These results are consistent with~\cite{VilaInsa2025}, which states that singular detection is achievable in the high \acrshort{snr} regime if and only if $\mathbf{X}_i$ and $\mathbf{X}_j$ span different column spaces.
			Within our framework, this property translates into~\eqref{eq:high_snr_energy_1} diverging, since $\mathbf{X}_i$ and $\mathbf{X}_j$ belong to different subspaces, whereas~\eqref{eq:high_snr_energy} is bounded due to both codewords spanning the same column space.
			
			A similar analysis can be performed for the low \acrshort{snr} regime, \ie $\power{\bsf{Z}}\to\infty$ (see~\eqref{eq:snr}).
			Using \acrshort{ldct} once again, we obtain the following second-order Taylor approximation:
			\begin{equation}
				\lim_{\power{\bsf{Z}}\to\infty} \overset{\infty}{\mathcal{D}}_{\mathrm{KL}}(i\Vert j) = \lim_{\power{\bsf{Z}}\to\infty} \tfrac{1}{\power{\bsf{Z}}^2} \smashoperator{\sum_{k=0}^{\nt-1}} \bigl(w_{i}^{(k)} - w_{j}^{(k)}\bigr)^2 \tfrac{1}{2}\int_{0}^{\mathrlap{\frac{1}{K}}} (\cyc{\bsf{H}}^{(k)}(K\sigma))^2
				\der\sigma.\label{eq:low_snr_energy}
			\end{equation}
			By~\cite[Prop.~6.A.3.]{Marshall2011}, sorting $\{w_{i}^{(k)}\}$ in descending order and $\{w_{j}^{(k)}\}$ in ascending order (or vice-versa) yields the maximum value of~\eqref{eq:low_snr_energy} asymptotically.
			Remarkably, this strategy is also optimal in~\eqref{eq:high_snr_energy} for the high \acrshort{snr} regime: since it is a Schur-convex function of $\{w_{i}^{(k)}/w_{j}^{(k)}\}$~\cite[Prop.~3.C.1.]{Marshall2011}, it is maximized when $\{w_{i}^{(k)}\}$ and $\{w_{j}^{(k)}\}$ are arranged in opposite orders~\cite[Prop.~6.A.1.]{Marshall2011}.
			This idea suggests a design approach on the power profiles of $\mathcal{X}$ based on the joint majorization properties of $\bigl\{\{w_i^{(p)}\},\{w_j^{(p')}\}\bigr\}$, for every pair $\mathbf{X}_i,\mathbf{X}_j\in\mathcal{X}$.
				
		\subsection{No \texorpdfstring{\acrshort{csi}}{} at the transmitter} \label{ssec:no_csit}
		
			We now consider the case in which the \acrshort{ue} is not aware of the channel basis nor spectrum.
			When the transmitter is not aware of the channel statistics, it is reasonable to set an agnostic precoding ($\mathbf{\Psi}_{i}\triangleq\id{\nt}$) and no energy modulation in the power-loading matrix ($\{(w_{i}^{(p)})^{\frac{1}{2}}=1\}_{p=0,\dots,N_i-1}$), so that no radiated direction is favored.
			For a clearer exposition, we set $N_i\triangleq\nt$ for all $\mathbf{X}_i\in\mathcal{X}$.
			This yields $\mathcal{X}$ to be a set of semi-unitary matrices:
			\begin{equation}
				\mathbf{X}_{i}\herm{\mathbf{X}_{i}}\triangleq\mathbf{P}_{i}, \quad \herm{\mathbf{X}_{i}}\mathbf{X}_{i}=\id{\nt},
			\end{equation}
			where $\mathbf{P}_{i}\in\complex^{K\times K}$ is the projection matrix onto $\mathrm{V}_i$.
			If every codeword in $\mathcal{X}$ spans a distinct subspace, this alphabet is known as a \textit{Grassmannian constellation}, since each $\mathrm{V}_i$ corresponds to a different point on the Grassmann manifold $\mathcal{G}(\nt,\complex^K)$~\cite{Zheng2002}.
			
			Under these premises, the eigendecomposition of the \acrshort{csm} of $\{\mathsf{y}_{i}(n)\}$ is
			\begin{equation}
				\mathbf{B}_i(\sigma)= \herm{\mathbf{F}_{K}} \herm{\mathbf{\Theta}}(\sigma)[\mathbf{U}_i,\mathbf{U}_{\overline{i}}], \quad \klcycmat{i}(\sigma) = \Diag\bigl(\cycmat{\bsf{H}}(K\sigma), \mathbf{0}_{K-\nt}\bigr) + \power{\bsf{Z}}\id{K}, \label{eq:csm_no_csit}
			\end{equation}
			where $\mathbf{U}_i\triangleq\mathbf{\Phi}_{i}\mathbf{W}_{i}\Ut$ and $\mathbf{U}_{\overline{i}}\in\complex^{K\times(K-\nt)}$ spans the orthogonal complement to $\mathrm{V}_i$ (its corresponding projector is $\mathbf{P}_{\overline{i}}\triangleq\mathbf{U}_{\overline{i}}\herm{\mathbf{U}}_{\overline{i}}$).
			It is clear that $\klcycmat{i}(\sigma)$ is common for all $\mathbf{X}_{i}$, so the logarithmic term in~\eqref{eq:kld_spectral} is null.
			By applying~\eqref{eq:csm_no_csit} onto~\eqref{eq:kld_spectral}, the resulting \acrshort{kld} is
			\begin{equation}
				\overset{\infty}{\mathcal{D}}_{\mathrm{KL}}(i\Vert j) = \int_{0}^{\mathrlap{\frac{1}{K}}} \der\sigma \trace\bigl[ \klcycmat{i}^{-1}(\sigma) \herm{[\mathbf{U}_j,\mathbf{U}_{\overline{j}}]} [\mathbf{U}_i,\mathbf{U}_{\overline{i}}] \klcycmat{i}(\sigma) \herm{[\mathbf{U}_i,\mathbf{U}_{\overline{i}}]} [\mathbf{U}_j,\mathbf{U}_{\overline{j}}] - \id{K}\bigr] .
			\end{equation}
			With simple algebraic manipulations, this is developed as
			\begin{subequations}
				\begin{align}
					&\overset{\infty}{\mathcal{D}}_{\mathrm{KL}}(i\Vert j) = \tfrac{1}{K} \int_{0}^{\mathrlap{1}} \trace[\mathbf{\Sigma}^{-1}(\lambda)\herm{\mathbf{U}_j}\mathbf{U}_i\mathbf{\Sigma}(\lambda)\herm{\mathbf{U}_i}\mathbf{U}_j]\der\lambda + \tfrac{1}{K}\lVert\herm{\mathbf{U}_{\overline{j}}}\mathbf{U}_{\overline{i}}\rVert_{\frob}^2 - 1 \label{eq:first}\\
					&\quad + \tfrac{1}{K}\int_{0}^{\mathrlap{1}}\der\lambda\bigl(\power{\bsf{Z}}\trace[\mathbf{U}_j\mathbf{\Sigma}^{-1}(\lambda)\herm{\mathbf{U}_j}\mathbf{P}_{\overline{i}}] + \tfrac{1}{\power{\bsf{Z}}}\trace[\herm{\mathbf{U}_{\overline{j}}}\mathbf{U}_i\mathbf{\Sigma}(\lambda)\herm{\mathbf{U}_i}\mathbf{U}_{\overline{j}}]\bigr) \label{eq:second}
				\end{align}
			\end{subequations}
			where $\mathbf{\Sigma}(\lambda)\triangleq\cycmat{\bsf{H}}(\lambda)+\power{\bsf{Z}}\id{\nt}$.
			The first term in~\eqref{eq:first} compares the signal subspaces between hypotheses $i$ and $j$, whereas the second term compares the noise subspaces.
			The two terms in~\eqref{eq:second} compare the crossed subspaces.
			
			This new expression allows to analyze the behavior of the \acrshort{kld} in the high and low \acrshort{snr} regimes as in \autoref{ssec:csit}.
			When $\power{\bsf{Z}}\to0$, the previous \acrshort{kld} becomes
			\begin{equation}
				\begin{multlined}
					\lim_{\power{\bsf{Z}}\to0}\overset{\infty}{\mathcal{D}}_{\mathrm{KL}}(i\Vert j) = \tfrac{1}{K} \int_{0}^{\mathrlap{1}} \trace[\cycmat{\bsf{H}}^{-1}(\lambda)\herm{\mathbf{U}_j}\mathbf{U}_i\cycmat{\bsf{H}}(\lambda)\herm{\mathbf{U}_i}\mathbf{U}_j]\der\lambda + \tfrac{1}{K}\lVert\herm{\mathbf{U}_{\overline{j}}}\mathbf{U}_{\overline{i}}\rVert_{\frob}^2 - 1\\
					+ \lim_{\power{\bsf{Z}}\to0} \tfrac{1}{K\power{\bsf{Z}}}\trace\bigl[\herm{\mathbf{U}_{\overline{j}}}\mathbf{U}_i\bigl({\textstyle\int_{0}^{\mathrlap{1}}\cycmat{\bsf{H}}(\lambda)\der\lambda}\bigr)\herm{\mathbf{U}_i}\mathbf{U}_{\overline{j}}\bigr].
				\end{multlined}
			\end{equation}
			Therefore, it will diverge as long as $\mathrm{V}_i\neq\mathrm{V}_j$ which is consistent with singular detection conditions in the high \acrshort{snr} regime~\cite{VilaInsa2025}.
			On the contrary, when $\power{\bsf{Z}}\to\infty$, the second-order Taylor approximation of the \acrshort{kld} is a function of correlations of spectra, similarly to what was observed in \autoref{ssec:csit} (see~\eqref{eq:low_snr_energy}).
			This implies high spectral correlation generally benefits detection when noise obscures the signal of interest, which is a fairly common behavior displayed at low \acrshort{snr} by classical techniques like the \textit{estimator-correlator}~\cite{Zhang2010}.
			
	\section{Concluding remarks}
	
		This work has presented a novel approach to analyze a noncoherent uplink massive \acrshort{mimo} system, by exploiting the cyclostationarity emerging in the joint treatment of the time and space domains.
		The use of large arrays has allowed the derivation of asymptotic representations of the received signals, from which relevant attributes for codeword detection become manifest.
		As a result of this theoretical study, a low-complexity implementation of the \acrshort{ml} detector has been inspired, which benefits from the low dimensional spectral correlation structures.
		
		A straightforward extension of the presented work may involve examining the downlink scenario using similar techniques as the ones described herein.
		Although the signal received by the \acrshort{ue} does not display a \acrshort{cs} nature, a representation in terms of the channel spectrum would be accessible (due to the massive array at the \acrshort{bs}) and provide insights on the problem, as well as potentially lead to simplified precoding strategies based on the \acrshort{fft} in the space domain.
		Another clear line of research may employ the asymptotic framework developed in this paper to design alphabets tailored to specific fading profiles.
		The limit expressions obtained in the high and low \acrshort{snr} regimes provide particularly useful design criteria.
		Additionally, the study of the spectral interplay between users in a multiple access setting may provide deep and novel perspectives on challenging problems.
				
	\backmatter
	
	\section*{Declarations}
	
		\bmhead{Availability of data and materials}
		
			Data sharing is not applicable to this article as no datasets were generated or analyzed during the current study.
			
		\bmhead{Competing interests}
		
			The authors declare that they have no competing interests.
	
		\bmhead{Funding}
		
			This work was (partially) funded by project MAYTE (PID2022-136512OB-C21) by MICIU/AEI/10.13039/501100011033 and ERDF/EU, grant 2021 SGR 01033 and grant 2022 FI SDUR 00164 by Departament de Recerca i Universitats de la Generalitat de Catalunya.
			
		\bmhead{Authors' contributions}
		
			MV-I: \textit{Original idea, research development, implementation of numerical simulations, preparation of the manuscript and revision.}
			JR: \textit{Original idea, research development, assessment of numerical simulations, preparation of the manuscript and revision.}
			
		\bmhead{Acknowledgments}
		
			Not applicable.
	
	\begin{appendices}
		
		\renewcommand{\theHequation}{AABB\arabic{equation}} 
		
		\section{Proof of \autoref{prop:channel}}\label{app:proof_channel_acf}
		
			From~\eqref{eq:channel_acf_entry}, recall the expression for the $k$th diagonal entry of $\covmat{\bsf{H}}(m)$, which is a Darboux sum.
			It converges to an integral by defining $\cyc{\bsf{H}}^{(k)}(\frac{l}{\nr})\triangleq\gamma_{k,l}$ and setting $\frac{l}{\nr}\to\lambda\in[0,1)$ and $\frac{1}{\nr}\to\der\lambda$:
			\begin{equation}
				\lim_{\nr\to\infty} [\covmat{\bsf{H}}(m)]_{k,k} = \int_{0}^{\mathrlap{1}}\cyc{\bsf{H}}^{(k)}(\lambda)\euler^{\imunit2\pi m\lambda}\der\lambda. \label{eq:wiener}
			\end{equation}
			Since $\{\mathring{\mathsf{h}}_k(n)\}$ is a \acrshort{wss} process, by the Wiener-Khinchin theorem~\cite[Sec.~9.2.2]{Schreier2010}, we know that $\cyc{\bsf{H}}^{(k)}(\lambda)$ is in fact its \acrshort{psd}.
			By defining $\cycmat{\bsf{H}}(\lambda)\triangleq\Diag(\{\cyc{\bsf{H}}^{(k)}(\lambda)\}_{k\in\{0,\dots,\nt-1\}})$ as in~\eqref{eq:channel_psd}, the resulting \acrshort{acf} from~\eqref{eq:channel_acf} can be expressed in the limit as
			\begin{equation}
				\lim_{\nr\to\infty} \cov{\bsf{H}}(n,m) = \smashoperator[l]{\sum_{p,q\in\integs}} \int_{0}^{\mathrlap{1}}[\Ut\cycmat{\bsf{H}}(\lambda)\herm{\Ut}]_{n+m-p\nt,n-q\nt}\euler^{\imunit2\pi(p-q)\lambda}\der\lambda.
			\end{equation}
		
		\section{Proof of \autoref{prop:cs}} \label{app:proof_cs}
		
			From~\eqref{eq:y_acf}, the \acrshort{acf} of $\{\mathsf{y}_{i}(n)\}$ is
			\begin{equation}
				\cov{i}(n,m) = \smashoperator[l]{\sum_{l,l\Prime=-\frac{\nr}{2}}^{\frac{\nr}{2}-1}}\smashoperator[r]{\sum_{k,k\Prime=0\vphantom{\frac{\nr}{2}}}^{\nt-1\vphantom{\frac{\nr}{2}}}}[\mathbf{X}_{i}]_{n+m-lK,k}[\mathbf{X}_{i}]_{n-l\Prime K,k\Prime}^*\expec\bigl[[\bsf{H}]_{k,l}[\bsf{H}]_{k\Prime,l\Prime}^*\bigr] +\power{\bsf{Z}}\delta_m.
			\end{equation}
			We have used the fact that $\bsf{H}$ and $\bsf{Z}$ are independent and the \acrshort{acf} of $\{[\vecrand{z}]_n\}$ from its covariance matrix in \autoref{ssec:model}.
			For $r\in\integs$,
			\begin{equation}
				\cov{i}(n+rK,m) = \smashoperator{\sum_{l,l\Prime,k,k\Prime}}[\mathbf{X}_{i}]_{n+m-(l-r)K,k}[\mathbf{X}_{i}]_{n-(l\Prime-r)K,k\Prime}^*\expec\bigl[[\bsf{H}]_{k,l}[\bsf{H}]_{k\Prime,l\Prime}^*\bigr] + \power{\bsf{Z}}\delta_m,
			\end{equation}
			Then, by setting $s=l-r$ and $s\Prime=l\Prime-r$, such that $s,s\Prime\in\{-\frac{\nr}{2}-r,\dots,\frac{\nr}{2}-1-r\}$,
			\begin{equation}
				\cov{i}(n+rK,m) = \smashoperator[l]{\sum_{s,s\Prime}}\smashoperator[r]{\sum_{k,k\Prime}}[\mathbf{X}_{i}]_{n+m-sK,k}[\mathbf{X}_{i}]_{n-s\Prime K,k\Prime}^*\expec\bigl[[\bsf{H}]_{k,s+r}[\bsf{H}]_{k\Prime,s\Prime+r}^*\bigr] + \power{\bsf{Z}}\delta_m.
			\end{equation}
			Finally, by the spatial stationarity of $\bsf{H}$ (refer to~\eqref{eq:joint_wss}), we can see that
			\begin{align}
				\cov{i}(n+rK,m) &= \smashoperator[l]{\sum_{s,s\Prime=-\frac{\nr}{2}}^{\frac{\nr}{2}-1}}\smashoperator{\sum_{k,k\Prime\vphantom{\frac{\nr}{2}}}}[\mathbf{X}_{i}]_{n+m-sK,k}[\mathbf{X}_{i}]_{n-s\Prime K,k\Prime}^*\expec\bigl[[\bsf{H}]_{k,s}[\bsf{H}]_{k\Prime,s\Prime}^*\bigr] + \power{\bsf{Z}}\delta_m \nonumber\\
				&\equiv \cov{i}(n,m). \label{eq:Ci}
			\end{align}
		
		\section{Proof of \autoref{prop:quadratic}} \label{app:proof_quadratic}
		
			Let $\Omega_{j}(n,m) \triangleq \lim_{\nr\to\infty} [\covmat{j}^{-1}]_{n+m,n}$ be the \textit{precision function} of $\{\mathsf{y}_{j}(n)\}$.
			By definition, we know that
			$[\covmat{j}\covmat{j}^{-1}]_{r,c} = [\covmat{j}^{-1}\covmat{j}]_{r,c} \triangleq \delta_{r-c}$.
			Therefore, in the limit of $\nr\to\infty$,
			\begin{equation}
				\lim_{\nr\to\infty} [\covmat{j}\covmat{j}^{-1}]_{r,c} = \smashoperator{\sum_{k\in\integs}} \infcov{j}(k,r-k)\Omega_{j}(c,k-c) \triangleq \delta_{r-c} \label{eq:cov_prec}
			\end{equation}
			must hold.
			From~\eqref{eq:ortho} and~\eqref{eq:eigen} we know that
			\begin{equation}
				\infcov{j}(c,r-c) = \int_{0}^{\mathrlap{1}} \phi_j(r,\lambda) \klcyc{j}(\lambda) \phi_{j}^*(c,\lambda) \der\lambda.
			\end{equation}
			It can then be verified that
			\begin{equation}
				\Omega_{j}(c,r-c) = \int_{0}^{\mathrlap{1}} \phi_{j}(r,\lambda) \klcyc{j}^{-1}(\lambda) \phi_{j}^*(c,\lambda) \der\lambda \label{eq:precision}
			\end{equation}
			does indeed fulfill~\eqref{eq:cov_prec}.
			
			We may express the quadratic form~\eqref{eq:quadratic} in the limit as
			\begin{equation}
				\overset{\infty}{\mathrm{Q}}_j(\{\mathsf{y}_{i}(n)\})  = \smashoperator{\sum_{r,c\in\integs}} \mathsf{y}_{i}^*(r) \Omega_{j}(c,r-c) \mathsf{y}_{i}(c)\der\sigma,
			\end{equation}
			where $\frac{1}{K\nr}\to\der\sigma$.
			Using~\eqref{eq:precision} and the \acrshort{kl} expansion of $\{\mathsf{y}_i(n)\}$~\eqref{eq:kl_exp}, the previous expression can be represented in the \acrshort{kl} spectral domain:
			\begin{equation}
				\overset{\infty}{\mathrm{Q}}_j(\{\mathsf{y}_{i}(n)\})
				= \der\sigma\iiint_{0}^{\mathrlap{1}} \der\mathring{\mathsf{y}}_i^*(\alpha) \frac{\langle\phi_{i},\phi_{j}\rangle(\alpha,\lambda) \langle\phi_{j},\phi_{i}\rangle(\lambda,\beta)}{\klcyc{j}(\lambda)} \der\mathring{\mathsf{y}}_i(\beta) \der\lambda,
			\end{equation}
			where
			\begin{equation}
				\langle\phi_a,\phi_b\rangle(\alpha,\beta)\triangleq\smashoperator{\sum_{n\in\integs}}\phi_a^*(n,\alpha)\phi_b(n,\beta)
			\end{equation} 
			is the inner product between two basis components.
			Since $\{\mathsf{y}_i(n)\}$ and $\{\mathsf{y}_j(n)\}$ are both \acrshort{cs} processes, it can particularized as
			\begin{equation}
				\langle\phi_a^{(k)},\phi_b^{(k\Prime)}\rangle(\sigma,\sigma\Prime) \triangleq \langle\phi_a,\phi_b\rangle\bigl(\sigma+\tfrac{k}{K},\sigma\Prime+\tfrac{k\Prime}{K}\bigr) = \delta(\sigma-\sigma\Prime) [\herm{\mathbf{B}}_a(\sigma)\mathbf{B}_b(\sigma\Prime)]_{k,k\Prime}.
			\end{equation}
			This yields
			\begin{equation}
				\overset{\infty}{\mathrm{Q}}_j(\{\mathsf{y}_{i}(n)\}) = \smashoperator[l]{\sum_{k,m}} \int_{0}^{\mathrlap{\frac{1}{K}}} [\herm{\mathbf{B}}_i(\sigma)\cycmat{j}^{-1}(\sigma)\mathbf{B}_i(\sigma)]_{k,m} \der\mathring{\mathsf{y}}_i^*\bigl(\sigma+\tfrac{k}{K}\bigr) \der\mathring{\mathsf{y}}_i\bigl(\sigma+\tfrac{m}{K}\bigr),
			\end{equation}
			where we have used~\eqref{eq:csm_eig} and the sifting property of the Dirac delta distribution.
			
		\section{Proof of \autoref{prop:csm}} \label{app:csm}
		
			Recall the cyclic spectrum from~\eqref{eq:csf}.
			We compute it for $\{\mathsf{y}_j(n)\}$ by using~\eqref{eq:C_inf}:
			\begin{multline}
				\cyc{j}^{(\frac{k}{K})}(\sigma) = \frac{1}{K} \smashoperator{\sum_{n\Prime=0\vphantom{\integs}}^{K-1}} \smashoperator[r]{\sum_{m\Prime\in\integs}} \power{\bsf{Z}}\delta_{m\Prime} \euler^{-\imunit2\pi(n\Prime\frac{k}{K}+m\Prime\sigma)} + \frac{1}{K} \smashoperator{\sum_{n=0}^{K-1}} \euler^{-\imunit2\pi n\frac{k}{K}}  \smashoperator{\sum_{m\in\integs}} \euler^{-\imunit2\pi m\sigma} \\
				\times  \smashoperator[l]{\sum_{l\in\integs}} \int_{0}^{\mathrlap{1}} [\mathbf{X}_{j}\Ut\cycmat{\bsf{H}}(\lambda)\herm{\Ut}\herm{\mathbf{X}_{j}}]_{n+m-lK,n}\euler^{\imunit2\pi l\lambda}\der\lambda.
			\end{multline}
			Notice we have set $l\Prime=0$ since $n\in\{0,\dots,K-1\}$.
			With $p=n+m-lK$, for $p\in\{0,\dots,K-1\}$, we have
			\begin{multline}
				\cyc{j}^{(\frac{k}{K})}(\sigma) = \tfrac{\power{\bsf{Z}}}{K} \smashoperator{\sum_{n\Prime=0}^{K-1}} \euler^{-\imunit2\pi n\Prime\frac{k}{K}} +\\
				\tfrac{1}{K}\smashoperator[l]{\sum_{n,p=0}^{K-1}} \int_{0}^{\mathrlap{1}} \euler^{-\imunit2\pi p\sigma} [\mathbf{X}_{j}\Ut\cycmat{\bsf{H}}(\lambda)\herm{\Ut}\herm{\mathbf{X}_{j}}]_{p,n} \euler^{\imunit2\pi n(\sigma-\frac{k}{K})} \sum_{l\in\integs} \euler^{\imunit2\pi l(\lambda-K\sigma)}\der\lambda.
			\end{multline}
			We may represent the previous expression more compactly by defining the steering vector $\boldsymbol{\theta}(\sigma)\triangleq\trans{[1,\euler^{\imunit2\pi\sigma},\dots,\euler^{\imunit2\pi\sigma(K-1)}]}$, which yields
			\begin{equation}
				\cyc{j}^{(\frac{k}{K})}(\sigma) = \tfrac{1}{K} \int_{0}^{\mathrlap{1}} \herm{\boldsymbol{\theta}}(\sigma) \mathbf{X}_{j}\Ut\cycmat{\bsf{H}}(\lambda)\herm{\Ut}\herm{\mathbf{X}_{j}} \boldsymbol{\theta}\bigl(\sigma-\tfrac{k}{K}\bigr) \delta(\lambda-K\sigma)\der\lambda +\power{\bsf{Z}}\delta_k.
			\end{equation}
			We have also converted the sum of exponentials with respect to $n\Prime$ into a Kronecker delta~\cite[Sec.~10.4.3]{Schreier2010}.
			Similarly, the sum with respect to $l$ has been treated as a Dirac delta distribution~\cite[Sec.~2.11]{Kennedy2013}.
			Using its sifting property we are left with
			\begin{equation}
				\cyc{j}^{(\frac{k}{K})}(\sigma) = \tfrac{1}{K} \herm{\boldsymbol{\theta}}(\sigma) \mathbf{X}_{j}\Ut\cycmat{\bsf{H}}(K\sigma)\herm{\Ut}\herm{\mathbf{X}_{j}} \boldsymbol{\theta}\bigl(\sigma-\tfrac{k}{K}\bigr) + \power{\bsf{Z}}\delta_k.
			\end{equation}
			
			With the cyclic spectrum, we are able to construct the \acrshort{csm} using~\eqref{eq:csm_def}:
			\begin{equation}
				[\cycmat{j}(\sigma)]_{r,c}= \tfrac{1}{K} \herm{\boldsymbol{\theta}}\bigl(\sigma+\tfrac{r}{K}\bigr) \mathbf{X}_{j}\Ut\cycmat{\bsf{H}}(K\sigma)\herm{\Ut}\herm{\mathbf{X}_{j}} \boldsymbol{\theta}\bigl(\sigma+\tfrac{c}{K}\bigr) + \power{\bsf{Z}}\delta_{r-c},
			\end{equation}
			for $r,c\in\{0,\dots,K-1\}$.
			We have used the fact that $\cycmat{\bsf{H}}(\lambda)$ is periodic in $\lambda$ with period 1, as seen in~\eqref{eq:channel_psd}, so that $\cycmat{\bsf{H}}(K\sigma+r)\equiv\cycmat{\bsf{H}}(K\sigma)$.
	\end{appendices}
	
	\bibliography{sn-bibliography}


\begin{thebibliography}{47}
\ifx \bisbn   \undefined \def \bisbn  #1{ISBN #1}\fi
\ifx \binits  \undefined \def \binits#1{#1}\fi
\ifx \bauthor  \undefined \def \bauthor#1{#1}\fi
\ifx \batitle  \undefined \def \batitle#1{#1}\fi
\ifx \bjtitle  \undefined \def \bjtitle#1{#1}\fi
\ifx \bvolume  \undefined \def \bvolume#1{\textbf{#1}}\fi
\ifx \byear  \undefined \def \byear#1{#1}\fi
\ifx \bissue  \undefined \def \bissue#1{#1}\fi
\ifx \bfpage  \undefined \def \bfpage#1{#1}\fi
\ifx \blpage  \undefined \def \blpage #1{#1}\fi
\ifx \burl  \undefined \def \burl#1{\textsf{#1}}\fi
\ifx \doiurl  \undefined \def \doiurl#1{\url{https://doi.org/#1}}\fi
\ifx \betal  \undefined \def \betal{\textit{et al.}}\fi
\ifx \binstitute  \undefined \def \binstitute#1{#1}\fi
\ifx \binstitutionaled  \undefined \def \binstitutionaled#1{#1}\fi
\ifx \bctitle  \undefined \def \bctitle#1{#1}\fi
\ifx \beditor  \undefined \def \beditor#1{#1}\fi
\ifx \bpublisher  \undefined \def \bpublisher#1{#1}\fi
\ifx \bbtitle  \undefined \def \bbtitle#1{#1}\fi
\ifx \bedition  \undefined \def \bedition#1{#1}\fi
\ifx \bseriesno  \undefined \def \bseriesno#1{#1}\fi
\ifx \blocation  \undefined \def \blocation#1{#1}\fi
\ifx \bsertitle  \undefined \def \bsertitle#1{#1}\fi
\ifx \bsnm \undefined \def \bsnm#1{#1}\fi
\ifx \bsuffix \undefined \def \bsuffix#1{#1}\fi
\ifx \bparticle \undefined \def \bparticle#1{#1}\fi
\ifx \barticle \undefined \def \barticle#1{#1}\fi
\bibcommenthead
\ifx \bconfdate \undefined \def \bconfdate #1{#1}\fi
\ifx \botherref \undefined \def \botherref #1{#1}\fi
\ifx \url \undefined \def \url#1{\textsf{#1}}\fi
\ifx \bchapter \undefined \def \bchapter#1{#1}\fi
\ifx \bbook \undefined \def \bbook#1{#1}\fi
\ifx \bcomment \undefined \def \bcomment#1{#1}\fi
\ifx \oauthor \undefined \def \oauthor#1{#1}\fi
\ifx \citeauthoryear \undefined \def \citeauthoryear#1{#1}\fi
\ifx \endbibitem  \undefined \def \endbibitem {}\fi
\ifx \bconflocation  \undefined \def \bconflocation#1{#1}\fi
\ifx \arxivurl  \undefined \def \arxivurl#1{\textsf{#1}}\fi
\csname PreBibitemsHook\endcsname

\bibitem[\protect\citeauthoryear{Gardner et~al.}{2006}]{Gardner2006}
\begin{barticle}
\bauthor{\bsnm{Gardner}, \binits{W.A.}},
\bauthor{\bsnm{Napolitano}, \binits{A.}},
\bauthor{\bsnm{Paura}, \binits{L.}}:
\batitle{Cyclostationarity: Half a century of research}.
\bjtitle{Signal Processing}
\bvolume{86}(\bissue{4}),
\bfpage{639}--\blpage{697}
(\byear{2006})
\end{barticle}
\endbibitem

\bibitem[\protect\citeauthoryear{Schkoda et~al.}{2014}]{Schkoda2014}
\begin{barticle}
\bauthor{\bsnm{Schkoda}, \binits{R.F.}},
\bauthor{\bsnm{Lund}, \binits{R.B.}},
\bauthor{\bsnm{Wagner}, \binits{J.R.}}:
\batitle{Clustering of cyclostationary signals with applications to climate
  station sitings, eliminations, and merges}.
\bjtitle{IEEE Journal of Selected Topics in Applied Earth Observations and
  Remote Sensing}
\bvolume{7}(\bissue{5}),
\bfpage{1754}--\blpage{1762}
(\byear{2014})
\end{barticle}
\endbibitem

\bibitem[\protect\citeauthoryear{Broszkiewicz-Suwaj
  et~al.}{2004}]{BroszkiewiczSuwaj2004}
\begin{barticle}
\bauthor{\bsnm{Broszkiewicz-Suwaj}, \binits{E.}},
\bauthor{\bsnm{Makagon}, \binits{A.}},
\bauthor{\bsnm{Weron}, \binits{R.}},
\bauthor{\bsnm{Wyłomańska}, \binits{A.}}:
\batitle{On detecting and modeling periodic correlation in financial data}.
\bjtitle{Physica A: Statistical Mechanics and its Applications}
\bvolume{336}(\bissue{1–2}),
\bfpage{196}--\blpage{205}
(\byear{2004})
\end{barticle}
\endbibitem

\bibitem[\protect\citeauthoryear{Hellbourg et~al.}{2011}]{Hellbourg2011}
\begin{barticle}
\bauthor{\bsnm{Hellbourg}, \binits{G.}},
\bauthor{\bsnm{Weber}, \binits{R.}},
\bauthor{\bsnm{Capdessus}, \binits{C.}},
\bauthor{\bsnm{Boonstra}, \binits{A.-J.}}:
\batitle{Cyclostationary approaches for spatial {RFI} mitigation in radio
  astronomy}.
\bjtitle{Comptes Rendus. Physique}
\bvolume{13}(\bissue{1}),
\bfpage{71}--\blpage{79}
(\byear{2011})
\end{barticle}
\endbibitem

\bibitem[\protect\citeauthoryear{Arora et~al.}{2008}]{Arora2008}
\begin{barticle}
\bauthor{\bsnm{Arora}, \binits{R.}},
\bauthor{\bsnm{Sethares}, \binits{W.A.}},
\bauthor{\bsnm{Bucklew}, \binits{J.A.}}:
\batitle{Latent periodicities in genome sequences}.
\bjtitle{IEEE Journal of Selected Topics in Signal Processing}
\bvolume{2}(\bissue{3}),
\bfpage{332}--\blpage{342}
(\byear{2008})
\end{barticle}
\endbibitem

\bibitem[\protect\citeauthoryear{Napolitano}{2013}]{Napolitano2013}
\begin{barticle}
\bauthor{\bsnm{Napolitano}, \binits{A.}}:
\batitle{Generalizations of cyclostationarity: A new paradigm for signal
  processing for mobile communications, radar, and sonar}.
\bjtitle{IEEE Signal Processing Magazine}
\bvolume{30}(\bissue{6}),
\bfpage{53}--\blpage{63}
(\byear{2013})
\end{barticle}
\endbibitem

\bibitem[\protect\citeauthoryear{Gardner}{1988}]{Gardner1988}
\begin{bbook}
\bauthor{\bsnm{Gardner}, \binits{W.A.}}:
\bbtitle{Statistical Spectral Analysis},
\bedition{1\textsuperscript{st}} edn.
\bsertitle{Prentice Hall Information and System Sciences Series}.
\bpublisher{Prentice-Hall, Inc.},
\blocation{Englewood Cliffs, NJ, USA}
(\byear{1988})
\end{bbook}
\endbibitem

\bibitem[\protect\citeauthoryear{Napolitano}{2016a}]{Napolitano2016}
\begin{barticle}
\bauthor{\bsnm{Napolitano}, \binits{A.}}:
\batitle{Cyclostationarity: New trends and applications}.
\bjtitle{Signal Processing}
\bvolume{120},
\bfpage{385}--\blpage{408}
(\byear{2016})
\end{barticle}
\endbibitem

\bibitem[\protect\citeauthoryear{Napolitano}{2016b}]{Napolitano2016a}
\begin{barticle}
\bauthor{\bsnm{Napolitano}, \binits{A.}}:
\batitle{Cyclostationarity: Limits and generalizations}.
\bjtitle{Signal Processing}
\bvolume{120},
\bfpage{323}--\blpage{347}
(\byear{2016})
\end{barticle}
\endbibitem

\bibitem[\protect\citeauthoryear{Antoni et~al.}{2004}]{Antoni2004}
\begin{barticle}
\bauthor{\bsnm{Antoni}, \binits{J.}},
\bauthor{\bsnm{Bonnardot}, \binits{F.}},
\bauthor{\bsnm{Raad}, \binits{A.}},
\bauthor{\bsnm{El~Badaoui}, \binits{M.}}:
\batitle{Cyclostationary modelling of rotating machine vibration signals}.
\bjtitle{Mechanical Systems and Signal Processing}
\bvolume{18}(\bissue{6}),
\bfpage{1285}--\blpage{1314}
(\byear{2004})
\end{barticle}
\endbibitem

\bibitem[\protect\citeauthoryear{Gardner}{1994}]{Gardner1994}
\begin{bbook}
\bauthor{\bsnm{Gardner}, \binits{W.A.}}:
\bbtitle{Cyclostationarity in Communications and Signal Processing}.
\bpublisher{IEEE Press},
\blocation{Piscataway, NJ, USA}
(\byear{1994})
\end{bbook}
\endbibitem

\bibitem[\protect\citeauthoryear{Shi et~al.}{2009}]{Shi2009}
\begin{barticle}
\bauthor{\bsnm{Shi}, \binits{M.}},
\bauthor{\bsnm{Bar-Ness}, \binits{Y.}},
\bauthor{\bsnm{Su}, \binits{W.}}:
\batitle{Revisiting the timing and frequency offset estimation based on
  cyclostationarity with new improved method}.
\bjtitle{IEEE Communications Letters}
\bvolume{13}(\bissue{7}),
\bfpage{537}--\blpage{539}
(\byear{2009})
\end{barticle}
\endbibitem

\bibitem[\protect\citeauthoryear{Horstmann et~al.}{2018}]{Horstmann2018}
\begin{barticle}
\bauthor{\bsnm{Horstmann}, \binits{S.}},
\bauthor{\bsnm{Ramirez}, \binits{D.}},
\bauthor{\bsnm{Schreier}, \binits{P.J.}}:
\batitle{Joint detection of almost-cyclostationary signals and estimation of
  their cycle period}.
\bjtitle{IEEE Signal Processing Letters}
\bvolume{25}(\bissue{11}),
\bfpage{1695}--\blpage{1699}
(\byear{2018})
\end{barticle}
\endbibitem

\bibitem[\protect\citeauthoryear{Tong et~al.}{1995}]{Tong1995}
\begin{barticle}
\bauthor{\bsnm{Tong}, \binits{L.}},
\bauthor{\bsnm{Xu}, \binits{G.}},
\bauthor{\bsnm{Hassibi}, \binits{B.}},
\bauthor{\bsnm{Kailath}, \binits{T.}}:
\batitle{Blind channel identification based on second-order statistics: a
  frequency-domain approach}.
\bjtitle{IEEE Transactions on Information Theory}
\bvolume{41}(\bissue{1}),
\bfpage{329}--\blpage{334}
(\byear{1995})
\end{barticle}
\endbibitem

\bibitem[\protect\citeauthoryear{Riba et~al.}{2010}]{Riba2010}
\begin{barticle}
\bauthor{\bsnm{Riba}, \binits{J.}},
\bauthor{\bsnm{Villares}, \binits{J.}},
\bauthor{\bsnm{Vázquez}, \binits{G.}}:
\batitle{A nondata-aided {SNR} estimation technique for multilevel modulations
  exploiting signal cyclostationarity}.
\bjtitle{IEEE Transactions on Signal Processing}
\bvolume{58}(\bissue{11}),
\bfpage{5767}--\blpage{5778}
(\byear{2010})
\end{barticle}
\endbibitem

\bibitem[\protect\citeauthoryear{Pries et~al.}{2018}]{Pries2018}
\begin{barticle}
\bauthor{\bsnm{Pries}, \binits{A.}},
\bauthor{\bsnm{Ramirez}, \binits{D.}},
\bauthor{\bsnm{Schreier}, \binits{P.J.}}:
\batitle{{LMPIT}-inspired tests for detecting a cyclostationary signal in noise
  with spatio–temporal structure}.
\bjtitle{IEEE Transactions on Wireless Communications}
\bvolume{17}(\bissue{9}),
\bfpage{6321}--\blpage{6334}
(\byear{2018})
\end{barticle}
\endbibitem

\bibitem[\protect\citeauthoryear{Chafii et~al.}{2023}]{Chafii2023}
\begin{barticle}
\bauthor{\bsnm{Chafii}, \binits{M.}},
\bauthor{\bsnm{Bariah}, \binits{L.}},
\bauthor{\bsnm{Muhaidat}, \binits{S.}},
\bauthor{\bsnm{Debbah}, \binits{M.}}:
\batitle{Twelve scientific challenges for {6G}: Rethinking the foundations of
  communications theory}.
\bjtitle{IEEE Communications Surveys \& Tutorials}
\bvolume{25}(\bissue{2}),
\bfpage{868}--\blpage{904}
(\byear{2023})
\end{barticle}
\endbibitem

\bibitem[\protect\citeauthoryear{Jing et~al.}{2016}]{Jing2016}
\begin{barticle}
\bauthor{\bsnm{Jing}, \binits{L.}},
\bauthor{\bsnm{Carvalho}, \binits{E.D.}},
\bauthor{\bsnm{Popovski}, \binits{P.}},
\bauthor{\bsnm{Martinez}, \binits{A.O.}}:
\batitle{Design and performance analysis of noncoherent detection systems with
  massive receiver arrays}.
\bjtitle{{IEEE} Transactions on Signal Processing}
\bvolume{64}(\bissue{19}),
\bfpage{5000}--\blpage{5010}
(\byear{2016})
\end{barticle}
\endbibitem

\bibitem[\protect\citeauthoryear{Marzetta}{2018}]{Marzetta2018}
\begin{bchapter}
\bauthor{\bsnm{Marzetta}, \binits{T.L.}}:
\bctitle{Spatially-stationary propagating random field model for massive {MIMO}
  small-scale fading}.
In: \bbtitle{2018 {IEEE} International Symposium on Information Theory
  ({ISIT})}.
\bpublisher{{IEEE}},
\blocation{Vail, CO, USA}
(\byear{2018})
\end{bchapter}
\endbibitem

\bibitem[\protect\citeauthoryear{Weichselberger
  et~al.}{2006}]{Weichselberger2006}
\begin{barticle}
\bauthor{\bsnm{Weichselberger}, \binits{W.}},
\bauthor{\bsnm{Herdin}, \binits{M.}},
\bauthor{\bsnm{Ozcelik}, \binits{H.}},
\bauthor{\bsnm{Bonek}, \binits{E.}}:
\batitle{A stochastic {MIMO} channel model with joint correlation of both link
  ends}.
\bjtitle{IEEE Transactions on Wireless Communications}
\bvolume{5}(\bissue{1}),
\bfpage{90}--\blpage{100}
(\byear{2006})
\end{barticle}
\endbibitem

\bibitem[\protect\citeauthoryear{Vilà-Insa and Riba}{2024}]{VilaInsa2024b}
\begin{bchapter}
\bauthor{\bsnm{Vilà-Insa}, \binits{M.}},
\bauthor{\bsnm{Riba}, \binits{J.}}:
\bctitle{A cyclostationary perspective on noncoherent {SIMO} communications}.
In: \bbtitle{2024 IEEE 25\textsuperscript{th} International Workshop on Signal
  Processing Advances in Wireless Communications (SPAWC)},
pp. \bfpage{71}--\blpage{75}.
\bpublisher{IEEE},
\blocation{Lucca, Italy}
(\byear{2024})
\end{bchapter}
\endbibitem

\bibitem[\protect\citeauthoryear{Riba and Vila}{2022}]{Riba2022}
\begin{barticle}
\bauthor{\bsnm{Riba}, \binits{J.}},
\bauthor{\bsnm{Vila}, \binits{M.}}:
\batitle{On infinite past predictability of cyclostationary signals}.
\bjtitle{{IEEE} Signal Processing Letters}
\bvolume{29},
\bfpage{647}--\blpage{651}
(\byear{2022})
\end{barticle}
\endbibitem

\bibitem[\protect\citeauthoryear{Vilà-Insa and Riba}{2024}]{VilaInsa2024}
\begin{botherref}
\oauthor{\bsnm{Vilà-Insa}, \binits{M.}},
\oauthor{\bsnm{Riba}, \binits{J.}}:
Asymptotic analysis of synchronous signal processing
(2024)
{\href{https://arxiv.org/abs/2403.18445}{{arXiv:2403.18445}}}
{[eess.SP]}
\end{botherref}
\endbibitem

\bibitem[\protect\citeauthoryear{Lapidoth}{2017}]{Lapidoth2017}
\begin{bbook}
\bauthor{\bsnm{Lapidoth}, \binits{A.}}:
\bbtitle{A Foundation in Digital Communication}.
\bpublisher{Cambridge University Press},
\blocation{Cambridge, United Kingdom}
(\byear{2017})
\end{bbook}
\endbibitem

\bibitem[\protect\citeauthoryear{Stoica and Nehorai}{1990}]{Stoica1990}
\begin{barticle}
\bauthor{\bsnm{Stoica}, \binits{P.}},
\bauthor{\bsnm{Nehorai}, \binits{A.}}:
\batitle{Performance study of conditional and unconditional
  direction-of-arrival estimation}.
\bjtitle{{IEEE} Transactions on Acoustics, Speech, and Signal Processing}
\bvolume{38}(\bissue{10}),
\bfpage{1783}--\blpage{1795}
(\byear{1990})
\end{barticle}
\endbibitem

\bibitem[\protect\citeauthoryear{Vilà-Insa and Riba}{2025}]{VilaInsa2025}
\begin{barticle}
\bauthor{\bsnm{Vilà-Insa}, \binits{M.}},
\bauthor{\bsnm{Riba}, \binits{J.}}:
\batitle{Singular detection in noncoherent communications}.
\bjtitle{IEEE Wireless Communications Letters}
\bvolume{14}(\bissue{4}),
\bfpage{1164}--\blpage{1168}
(\byear{2025})
\end{barticle}
\endbibitem

\bibitem[\protect\citeauthoryear{Feng et~al.}{2022}]{Feng2022}
\begin{barticle}
\bauthor{\bsnm{Feng}, \binits{R.}},
\bauthor{\bsnm{Wang}, \binits{C.-X.}},
\bauthor{\bsnm{Huang}, \binits{J.}},
\bauthor{\bsnm{Gao}, \binits{X.}},
\bauthor{\bsnm{Salous}, \binits{S.}},
\bauthor{\bsnm{Haas}, \binits{H.}}:
\batitle{Classification and comparison of massive {MIMO} propagation channel
  models}.
\bjtitle{IEEE Internet of Things Journal}
\bvolume{9}(\bissue{23}),
\bfpage{23452}--\blpage{23471}
(\byear{2022})
\end{barticle}
\endbibitem

\bibitem[\protect\citeauthoryear{Heath~Jr. and Lozano}{2018}]{Heath2018}
\begin{bbook}
\bauthor{\bsnm{Heath~Jr.}, \binits{R.W.}},
\bauthor{\bsnm{Lozano}, \binits{A.}}:
\bbtitle{Foundations of MIMO Communication},
\bedition{1\textsuperscript{st}} edn.
\bpublisher{Cambridge University Press},
\blocation{Cambridge, United Kingdom}
(\byear{2018})
\end{bbook}
\endbibitem

\bibitem[\protect\citeauthoryear{Sanguinetti et~al.}{2020}]{Sanguinetti2020}
\begin{barticle}
\bauthor{\bsnm{Sanguinetti}, \binits{L.}},
\bauthor{\bsnm{Bjornson}, \binits{E.}},
\bauthor{\bsnm{Hoydis}, \binits{J.}}:
\batitle{Toward massive {MIMO} 2.0: Understanding spatial correlation,
  interference suppression, and pilot contamination}.
\bjtitle{{IEEE} Transactions on Communications}
\bvolume{68}(\bissue{1}),
\bfpage{232}--\blpage{257}
(\byear{2020})
\end{barticle}
\endbibitem

\bibitem[\protect\citeauthoryear{Gray}{2006}]{Gray2005}
\begin{barticle}
\bauthor{\bsnm{Gray}, \binits{R.M.}}:
\batitle{Toeplitz and circulant matrices: A review}.
\bjtitle{Foundations and Trends® in Communications and Information Theory}
\bvolume{2}(\bissue{3}),
\bfpage{155}--\blpage{239}
(\byear{2006})
\end{barticle}
\endbibitem

\bibitem[\protect\citeauthoryear{Vaidyanathan}{2008}]{Vaidyanathan2008}
\begin{bbook}
\bauthor{\bsnm{Vaidyanathan}, \binits{P.P.}}:
\bbtitle{The Theory of Linear Prediction}.
\bsertitle{Synthesis Lectures on Signal Processing},
vol. \bseriesno{3}.
\bpublisher{Springer},
\blocation{San Rafael, CA, USA}
(\byear{2008})
\end{bbook}
\endbibitem

\bibitem[\protect\citeauthoryear{Zhang et~al.}{2020}]{Zhang2020}
\begin{barticle}
\bauthor{\bsnm{Zhang}, \binits{J.}},
\bauthor{\bsnm{Bjornson}, \binits{E.}},
\bauthor{\bsnm{Matthaiou}, \binits{M.}},
\bauthor{\bsnm{Ng}, \binits{D.W.K.}},
\bauthor{\bsnm{Yang}, \binits{H.}},
\bauthor{\bsnm{Love}, \binits{D.J.}}:
\batitle{Prospective multiple antenna technologies for beyond {5G}}.
\bjtitle{IEEE Journal on Selected Areas in Communications}
\bvolume{38}(\bissue{8}),
\bfpage{1637}--\blpage{1660}
(\byear{2020})
\end{barticle}
\endbibitem

\bibitem[\protect\citeauthoryear{Ramirez et~al.}{2015}]{Ramirez2015}
\begin{barticle}
\bauthor{\bsnm{Ramirez}, \binits{D.}},
\bauthor{\bsnm{Schreier}, \binits{P.J.}},
\bauthor{\bsnm{Via}, \binits{J.}},
\bauthor{\bsnm{Santamaria}, \binits{I.}},
\bauthor{\bsnm{Scharf}, \binits{L.L.}}:
\batitle{Detection of multivariate cyclostationarity}.
\bjtitle{{IEEE} Transactions on Signal Processing}
\bvolume{63}(\bissue{20}),
\bfpage{5395}--\blpage{5408}
(\byear{2015})
\end{barticle}
\endbibitem

\bibitem[\protect\citeauthoryear{Schreier and Scharf}{2010}]{Schreier2010}
\begin{bbook}
\bauthor{\bsnm{Schreier}, \binits{P.J.}},
\bauthor{\bsnm{Scharf}, \binits{L.L.}}:
\bbtitle{Statistical Signal Processing of Complex-Valued Data}.
\bpublisher{Cambridge University Press},
\blocation{Cambridge, United Kingdom}
(\byear{2010})
\end{bbook}
\endbibitem

\bibitem[\protect\citeauthoryear{Kennedy and Sadeghi}{2013}]{Kennedy2013}
\begin{bbook}
\bauthor{\bsnm{Kennedy}, \binits{R.A.}},
\bauthor{\bsnm{Sadeghi}, \binits{P.}}:
\bbtitle{Hilbert Space Methods in Signal Processing}.
\bpublisher{Cambridge University Press},
\blocation{Cambridge, United Kingdom}
(\byear{2013})
\end{bbook}
\endbibitem

\bibitem[\protect\citeauthoryear{Widom}{1974}]{Widom1974}
\begin{barticle}
\bauthor{\bsnm{Widom}, \binits{H.}}:
\batitle{Asymptotic behavior of block {Toeplitz} matrices and determinants}.
\bjtitle{Advances in Mathematics}
\bvolume{13}(\bissue{3}),
\bfpage{284}--\blpage{322}
(\byear{1974})
\end{barticle}
\endbibitem

\bibitem[\protect\citeauthoryear{Hlawatsch}{2011}]{Hlawatsch2011}
\begin{bbook}
\bauthor{\bsnm{Hlawatsch}, \binits{F.}}:
\bbtitle{Wireless Communications over Rapidly Time-Varying Channels}.
\bpublisher{Elsevier Science \& Technology},
\blocation{San Diego, CA, USA}
(\byear{2011})
\end{bbook}
\endbibitem

\bibitem[\protect\citeauthoryear{Levy}{2008}]{Levy2008}
\begin{bbook}
\bauthor{\bsnm{Levy}, \binits{B.C.}}:
\bbtitle{Principles of Signal Detection and Parameter Estimation}.
\bpublisher{Springer},
\blocation{New York, NY, USA}
(\byear{2008})
\end{bbook}
\endbibitem

\bibitem[\protect\citeauthoryear{Yang et~al.}{2023}]{Yang2023}
\begin{botherref}
\oauthor{\bsnm{Yang}, \binits{T.}},
\oauthor{\bsnm{Barzegar~Khalilsarai}, \binits{M.}},
\oauthor{\bsnm{Haghighatshoar}, \binits{S.}},
\oauthor{\bsnm{Caire}, \binits{G.}}:
Structured channel covariance estimation from limited samples for large antenna
  arrays.
EURASIP Journal on Wireless Communications and Networking
\textbf{2023}(1)
(2023)
\end{botherref}
\endbibitem

\bibitem[\protect\citeauthoryear{Kim et~al.}{2008}]{Kim2008}
\begin{barticle}
\bauthor{\bsnm{Kim}, \binits{T.}},
\bauthor{\bsnm{Bengtsson}, \binits{M.}},
\bauthor{\bsnm{Larsson}, \binits{E.}},
\bauthor{\bsnm{Skoglund}, \binits{M.}}:
\batitle{Transactions letters - {Combining} long-term and low-rate short-term
  channel state information over correlated {MIMO} channels}.
\bjtitle{IEEE Transactions on Wireless Communications}
\bvolume{7}(\bissue{7}),
\bfpage{2409}--\blpage{2414}
(\byear{2008})
\end{barticle}
\endbibitem

\bibitem[\protect\citeauthoryear{Kay}{1988}]{Kay1988}
\begin{bbook}
\bauthor{\bsnm{Kay}, \binits{S.M.}}:
\bbtitle{Modern Spectral Estimation}.
\bsertitle{Prentice-Hall Signal Processing Series}.
\bpublisher{Prentice Hall},
\blocation{Englewood Cliffs, NJ, USA}
(\byear{1988})
\end{bbook}
\endbibitem

\bibitem[\protect\citeauthoryear{Vetterli et~al.}{2014}]{Vetterli2014}
\begin{bbook}
\bauthor{\bsnm{Vetterli}, \binits{M.}},
\bauthor{\bsnm{Kovačević}, \binits{J.}},
\bauthor{\bsnm{Goyal}, \binits{V.K.}}:
\bbtitle{Foundations of Signal Processing}.
\bpublisher{Cambridge University Press},
\blocation{Cambridge, United Kingdom}
(\byear{2014})
\end{bbook}
\endbibitem

\bibitem[\protect\citeauthoryear{Mezzadri}{2007}]{Mezzadri2007}
\begin{barticle}
\bauthor{\bsnm{Mezzadri}, \binits{F.}}:
\batitle{How to generate random matrices from the classical compact groups}.
\bjtitle{Notices of the American Mathematical Society}
\bvolume{54}(\bissue{5}),
\bfpage{592}--\blpage{604}
(\byear{2007})
\end{barticle}
\endbibitem

\bibitem[\protect\citeauthoryear{Choudary and Niculescu}{2014}]{Choudary2014}
\begin{bbook}
\bauthor{\bsnm{Choudary}, \binits{A.D.R.}},
\bauthor{\bsnm{Niculescu}, \binits{C.P.}}:
\bbtitle{Real Analysis on Intervals},
\bedition{1}st edn.
\bsertitle{Springer eBook Collection}.
\bpublisher{Springer},
\blocation{New Delhi, India}
(\byear{2014})
\end{bbook}
\endbibitem

\bibitem[\protect\citeauthoryear{Marshall et~al.}{2011}]{Marshall2011}
\begin{bbook}
\bauthor{\bsnm{Marshall}, \binits{A.W.}},
\bauthor{\bsnm{Olkin}, \binits{I.}},
\bauthor{\bsnm{Arnold}, \binits{B.C.}}:
\bbtitle{Inequalities: Theory of Majorization and Its Applications}.
\bpublisher{Springer},
\blocation{New York, NY, USA}
(\byear{2011})
\end{bbook}
\endbibitem

\bibitem[\protect\citeauthoryear{Zheng and Tse}{2002}]{Zheng2002}
\begin{barticle}
\bauthor{\bsnm{Zheng}, \binits{L.}},
\bauthor{\bsnm{Tse}, \binits{D.N.C.}}:
\batitle{Communication on the {Grassmann} manifold: a geometric approach to the
  noncoherent multiple-antenna channel}.
\bjtitle{{IEEE} Transactions on Information Theory}
\bvolume{48}(\bissue{2}),
\bfpage{359}--\blpage{383}
(\byear{2002})
\end{barticle}
\endbibitem

\bibitem[\protect\citeauthoryear{Zhang et~al.}{2010}]{Zhang2010}
\begin{barticle}
\bauthor{\bsnm{Zhang}, \binits{W.}},
\bauthor{\bsnm{Poor}, \binits{H.V.}},
\bauthor{\bsnm{Quan}, \binits{Z.}}:
\batitle{Frequency-domain correlation: An asymptotically optimum approximation
  of quadratic likelihood ratio detectors}.
\bjtitle{IEEE Transactions on Signal Processing}
\bvolume{58}(\bissue{3}),
\bfpage{969}--\blpage{979}
(\byear{2010})
\end{barticle}
\endbibitem

\end{thebibliography}

\end{document}